\documentclass[12pt]{amsart}
\usepackage{amsmath,amsthm,amsfonts,amssymb,latexsym,verbatim,bbm,hyperref,mathrsfs}
\usepackage{mathtools}
\usepackage[shortlabels]{enumitem}
\usepackage{subcaption}
\usepackage{thmtools}
\usepackage[capitalize]{cleveref}
\usepackage{afterpage}

\DeclareRobustCommand{\SkipTocEntry}[5]{}

\allowdisplaybreaks

\setlist{itemsep=.5\baselineskip,topsep=.5\baselineskip}

\numberwithin{equation}{section}
\theoremstyle{plain}
\newtheorem{theorem}{Theorem}[section]

\newtheorem{lemma}[theorem]{Lemma}
\newtheorem{proposition}[theorem]{Proposition}

\newtheorem{example}[theorem]{Example}
\newtheorem{definition}[theorem]{Definition}

\newtheorem{remark}[theorem]{Remark}

\newtheorem{question}[theorem]{Question}
\newtheorem{problem}[theorem]{Problem}

\newcommand{\iso}{\cong}
\newcommand{\arr}{\rightarrow}

\newcommand{\ang}[1]{\langle #1 \rangle}
\newcommand{\R}{\mathbb{R}}
\newcommand{\C}{\mathbb{C}}
\newcommand{\Z}{\mathbb{Z}}
\newcommand{\N}{\mathbb{N}}

\newcommand{\eps}{\epsilon}

\newcommand{\Id}{1}

\newcommand{\mcA}{\mathcal{A}}
\newcommand{\mcB}{\mathcal{B}}
\newcommand{\mcC}{\mathcal{C}}
\newcommand{\mcD}{\mathcal{D}}

\newcommand{\mcF}{\mathcal{F}}
\newcommand{\mcG}{\mathcal{G}}
\newcommand{\mcH}{\mathcal{H}}
\newcommand{\mcI}{\mathcal{I}}
\newcommand{\mcJ}{\mathcal{J}}

\newcommand{\mcL}{\mathcal{L}}

\newcommand{\mcN}{\mathcal{N}}

\newcommand{\mcP}{\mathcal{P}}

\newcommand{\mcR}{\mathcal{R}}
\newcommand{\mcS}{\mathcal{S}}

\newcommand{\mcV}{\mathcal{V}}
\newcommand{\mcW}{\mathcal{W}}
\newcommand{\mcX}{\mathcal{X}}
\newcommand{\mcY}{\mathcal{Y}}

\newcommand{\msA}{\mathscr{A}}
\newcommand{\msB}{\mathscr{B}}

\newcommand{\msU}{\mathscr{U}}
\newcommand{\msR}{\mathscr{R}}
\newcommand{\msC}{\mathscr{C}}

\newcommand{\abs}[1]{\lvert #1 \rvert}

\newcommand{\tr}{\text{tr}}
\newcommand{\wtd}{\widetilde}
\newcommand{\norm}[1]{\lVert #1 \rVert}

\usepackage{xcolor} 

\usepackage[top=1.5in,bottom=1.5in,left=1.5in,right=1.5in]{geometry}

\DeclareMathOperator{\RE}{RE}
\DeclareMathOperator{\coRE}{coRE}

\DeclareMathOperator{\MIP}{MIP}

\renewcommand{\Id}{\mathbbm{1}}
\usepackage{braket}
\usepackage{tikz-cd}

\title[The NPA hierarchy does not always attain $\omega_{qc}(\mcG)$]{The NPA hierarchy does not always attain\\ the commuting operator value}

\author[Fanizza]{Marco Fanizza}
\author[Kroell]{Larissa Kroell}
\author[Mehta]{Arthur Mehta}
\author[Paddock]{Connor Paddock}
\author[Rochette]{Denis Rochette}
\author[Slofstra]{William Slofstra}
\author[Zhao]{Yuming Zhao}

\address{Marco Fanizza, Inria, Télécom Paris - LTCI, Institut Polytechnique de Paris and QMATH, Department of Mathematical Sciences, University of Copenhagen}
\email{marco.fanizza@inria.fr}

\address{Larissa Kroell, Department of Pure Mathematics, University of Waterloo}
\email{lkroell@uwaterloo.ca}

\address{Arthur Mehta, Department of Mathematics and Statistics, University of Ottawa}
\email{amehta2@Uottawa.ca}

\address{Connor Paddock, Department of Mathematics and Statistics, University of Ottawa}
\email{cpaulpad@uottawa.ca}

\address{William Slofstra, Institute for Quantum Computing and Department of Pure Mathematics, University of Waterloo}
\email{william.slofstra@uwaterloo.ca}

\address{Yuming Zhao, QMATH, Department of Mathematical Sciences, University of Copenhagen}
\email{yuming@math.ku.dk}

\begin{document}

\begin{abstract}
We show that it is undecidable to determine whether the commuting operator value of a nonlocal game is strictly greater than $\frac{1}{2}$. Specifically, there is a computable mapping from Turing machines to boolean constraint system (BCS) nonlocal games in which the halting property of the machine is encoded as a decision problem for the commuting operator value of the game. As a corollary, there is a BCS game for which the value of the Navascu{\'e}s-Pironio-Ac{\'\i}n (NPA) hierarchy does not attain the commuting operator value at any finite level.
\end{abstract}

\maketitle

\section{Introduction}
The computability of the quantum and commuting operator value of a nonlocal game has become an important topic in quantum information theory. Early hardness results for these nonlocal game values include
\cite{KKM+11, IKM09, IV12, Vid16,Vid20eratum, Ji16, Ji17, NV18a, FJVY19}, and notably \cite{S20} which gave the first undecidability result. In particular, deciding whether the commuting-operator value of a nonlocal game $\omega_{qc}(\mathcal{G})$ is equal to $1$ or is strictly below $1$ is $\coRE$-complete \cite{S20}. Subsequent work established that deciding if the quantum value $\omega_q(\mcG)$ is equal to $1$ or strictly below $1$ is also undecidable \cite{S19, MNY20}. Subsequently, the celebrated $\MIP^* = \RE$ result shows that the problem of deciding whether the quantum value $\omega_q(\mathcal{G})$ equals $1$ or is at most $1/2$ is $\RE$-complete \cite{JNVWY21}. The recent $\MIP^{co} = \coRE$ theorem shows that the same problem for the commuting operator value of the game is $\coRE$-complete \cite{MIPcocoRE}.

Closely tied to these undecidability results are outer approximation algorithms like the NPA hierarchy \cite{NPA07}. The NPA hierarchy, due to Navascu{\'e}s, Pironio, and Ac{\'i}n, is a hierarchy of semidefinite programs (SDPs) that provide a non-increasing sequence of upper bounds 
$\{\omega^{(k)}_{\mathrm{npa}}(\mcG)\}_{k\in \N}$ converging to the commuting operator value $\omega_{qc}(\mcG)$ \cite{NPA08}. A corollary of $\MIP^*=\RE$ employs the NPA hierarchy to conclude that there are nonlocal games for which $\omega_q(\mcG)$ is strictly less than $\omega_{qc}(\mcG)$; indeed, if the quantum and commuting operator values were always equal, then it would be possible to decide if $\omega_q(\mcG)=1$  or $\omega_{qc}(\mcG) \leq 1/2$ by computing the values of the NPA hierarchy and simultaneously searching through finite-dimensional strategies, until either the lower bounds are greater than $1/2$, or the upper bounds from the NPA hierarchy fall below $1$.

Beyond the connection to undecidability, the NPA hierarchy is a widely used tool for analyzing nonlocality across a range of settings in quantum information science, see for instance \cite{Tavakoli2024}. As a recent example, the NPA hierarchy has been employed to examine the quantum soundness of the cryptographic compiler of \cite{KLVY23}, resulting in quantitative bounds on the value of compiled games~\cite{NZ23a, CMM+24, CFNZ25, MPW25, KPMO+25}. One approach, outlined in \cite{CFNZ25}, bounds the compiled value of $\mcG$ within negligible of $\omega_{qc}(\mcG)$ for any game satisfying $\omega^{(k)}_{\mathrm{npa}}(\mcG) = \omega_{qc}(\mcG)$ at some level $k$. Given the important role of the NPA hierarchy in studying nonlocality, one is led to a natural question:

\begin{question}\label{que:npa_attain}
    For every nonlocal game $\mcG$, does there exist $k \in \mathbb{N}$ such that $$\omega^{(k)}_{\mathrm{npa}}(\mcG) = \omega_{qc}(\mcG)\;?$$ 
\end{question}
Interestingly, both a positive and a negative answer of \Cref{que:npa_attain} appear to be consistent with the $\coRE$-completeness of deciding if $\omega_{qc}(\mathcal{G})=1$ and the $\MIP^{co} = \coRE$ theorem. 
Indeed, suppose that for every game $\mcG$, there is a level $k(\mcG)$ such that $\omega^{(k(\mcG))}_{npa}(\mcG) = \omega_{qc}(\mcG)$. If we could compute the function $k(\mcG)$, then we could use this function to decide if the value of $\omega_{qc}(\mathcal{G})=1$ by computing $\omega^{(k(\mcG))}_{npa}(\mcG)$. So the $\coRE$-completeness results mentioned above imply that if such a function exists, it must be uncomputable. However, this argument does not imply that such a function does not exist. 

One might think that the reason the previous work mentioned above fails to provide insight into \cref{que:npa_attain} is because they are decision problems with perfect completeness. In the yes instances, $\omega_{qc}(\mcG)=1$, and we will have $\omega^{(k)}_{npa}(\mcG)=1$ for all $k$.
To rectify this, one might be tempted to consider the following problem, which does not have perfect completeness:
\begin{equation*}
    \begin{minipage}{.65\textwidth}
    Given a nonlocal game $\mcG$ decide if $\omega_{qc}(\mcG) \geq \frac{1}{2}$.
    \end{minipage}
\end{equation*}
The NPA algorithm shows that this problem is in $\coRE$, and a reduction 
from the main result of \cite{S20} shows that it is $\coRE$-complete, so this problem can also be used to show that $k(\mcG)$ is uncomputable if it exists. However, we do not get any new insights into \cref{que:npa_attain}.

Somewhat surprisingly, a slight alteration to the above decision problem makes all the difference. We consider the following decision problem:
\begin{equation*}
    \tag{QC-Strict}
    \begin{minipage}{.65\textwidth}
    Given a nonlocal game $\mcG$, decide if $\omega_{qc}(\mcG)>\frac{1}{2}$.
    \end{minipage}
\end{equation*}
Similar to what happens with the non-strict version, the $\MIP^{co}=\coRE$ theorem implies that (QC-Strict) is $\coRE$-hard, since an oracle for (QC-Strict) can be used to decide the promise problem $\omega_{qc}(\mcG) =1$ or $\omega_{qc}(\mcG) \leq 1/2$. However, unlike the non-strict version, this problem has a strong connection to \cref{que:npa_attain}.
If the answer to \cref{que:npa_attain} is yes, then (QC-Strict) is contained in $\coRE$, since for each level $k$ of the hierarchy, the problem of determining if $\omega^{(k)}_{\mathrm{npa}}(\mcG) \leq 1/2$ is decidable\footnote{The question of whether the optimal value of an SDP is $\leq 1/2$ can be solved exactly (e.g.~without worrying about numerical precision) using Tarski quantifier elimination.}. Our main result is that (QC-Strict) is $\RE$-hard:
\begin{theorem}\label{thm:introI}
    It is $\RE$-hard to determine whether the commuting-operator value of a nonlocal game is strictly greater than $1/2$.
\end{theorem}
Because (QC-Strict) is $\RE$-hard, it cannot be contained in $\coRE$, so the answer to \cref{que:npa_attain} cannot be yes.
\begin{theorem}\label{thm:main_npa}
    There exists a nonlocal game $\mcG$ for which $\omega^{(k)}_{\mathrm{npa}}(\mcG)>\omega_{qc}(\mcG)$ for all $k\in \N$.
\end{theorem}


\subsection{Techniques}\label{subsec:Techniques}

Our approach builds on the mathematical framework of \cite{MSZ23}, which studies the hardness of deciding if an element $p \in \mathbb{C}\mathbb{Z}^{*n}_m \otimes \mathbb{C}\mathbb{Z}^{*n}_m$ is positive. Here, $\mathbb{C}\mathbb{Z}^{*n}_m \otimes \mathbb{C}\mathbb{Z}^{*n}_m$ is the tensor product of the group algebras $\C\Z_m^{*n}$, which in quantum information abstractly models a bipartite measurement scenario with $n$ settings and $m$ outcomes. Elements $p \in \mathbb{C}\mathbb{Z}^{*n}_m \otimes \mathbb{C}\mathbb{Z}^{*n}_m$ are sometimes called $*$-polynomials, and a $*$-polynomial $p$ is positive if $\pi(p)$ is a positive semi-definite operator for all $*$-representations $\pi: \mathbb{C}\mathbb{Z}^{*n}_m \otimes \mathbb{C}\mathbb{Z}^{*n}_m \rightarrow \mcB(\mcH)$, where $\mcB(\mcH)$ is the bounded operators on the Hilbert space $\mcH$. Via the well-known Gelfand-Naimark-Segal (GNS) construction, one can show that $p$ is positive if and only if $\phi(p) \geq 0$ for all states $\phi: \mathbb{C}\mathbb{Z}^{*n}_m \otimes \mathbb{C}\mathbb{Z}^{*n}_m \rightarrow \C$. 

Given a two-player nonlocal game $\mcG$ with input sets of size $n$ and output sets of size $m$, there is a $*$-polynomial $\Phi_{\mcG} \in \C \Z_m^{*n} \otimes \C \Z_m^{*n}$ 
called the game functional (or game polynomial) of $\mcG$. If $\mcS$ is a commuting operator strategy for $\mcG$ using state $\ket{\psi}$, in which Alice uses observables $X_1,\ldots,X_n$ and
Bob uses observables $Y_1,\ldots,Y_n$, then the winning probability of $\mcG$ with strategy $\mcS$ is determined by the value $\braket{\psi|\Phi_{\mcG}(X_1,\ldots,X_n,Y_1,\ldots,Y_n)|\psi}.$ Moreover, the map sending $*$-polynomials $q$ to $\braket{\psi|q(X_1,\ldots,X_n,Y_1,\ldots,Y_n)|\psi}$ defines a state on $\mathbb{C}\mathbb{Z}^{*n}_m \otimes \mathbb{C}\mathbb{Z}^{*n}_m$. Conversely, by the GNS construction any state on $\mathbb{C}\mathbb{Z}^{*n}_m \otimes \mathbb{C}\mathbb{Z}^{*n}_m$ arises from some commuting operator strategy $\mathcal{S}$. Consequently, $\omega_{qc}(\mcG)$ is the smallest $\lambda \in \R$ such that $\lambda - \Phi_{\mcG}$ is positive. In particular, if $\alpha$ is a $*$-polynomial of the form $1/2 - \Phi_{\mcG}$ for some $\mcG$, then $\alpha$ is positive if and only if $\omega_{qc}(\mcG) \leq 1/2$.

In \cite{MSZ23} it is shown that there is a computable mapping from Turing machines $M$ to $*$-polynomials $\alpha_M$ such that $\alpha_M$ is positive if and only if $M$ does not halt. Hence to prove \Cref{thm:introI}, it would suffice to adapt the proof to show that there is a computable mapping $M \mapsto \alpha_M$ from Turing machines to $\C\Z_m^{*n} \otimes \C\Z_m^{*n}$, such that $\alpha_M$ is positive if and only if $M$ does not halt, and $\alpha_M = 1/2 - \Phi_{\mcG}$ for some two-player nonlocal game $\mcG$.

This approach presents some challenges. To understand why, we need to delve into the details of the computable mapping $M \mapsto \alpha_M$ in \cite{MSZ23}. The algorithm starts by writing down polynomial relations $r_1=0,\ldots,r_k=0$, and defines $\alpha_M := \sum_{i=1}^k r_i^* r_i - x$ for some positive term $x$. The relations depend on the Turing machine $M$ (for instance, their degree is $2^{|M|}$), and the resulting $\alpha_M$ does not obviously correspond to a game functional (which must have degree two). The issue of degree would not be a significant barrier if $r_1, \ldots, r_k$ were all group relations. This is because any finitely presented group can be embedded into a so-called solution group~\cite{S19}, whose group algebra (called the LCS algebra) corresponds to perfect strategies for a linear constraint system (LCS) games \cite{CLS17}. However, the relations $r_1, \ldots, r_k$ are not only group relations. To overcome this issue, we establish a more general embedding theorem for BCS algebras.

BCS algebras are finitely presented $*$-algebras associated to boolean constraint systems (BCS) nonlocal games \cite{CM24,PS23,AKS17,CM14,J13}. BCS algebras generalize LCS algebras. In particular, they have greater flexibility among their algebraic relations. For example, the boolean constraints $x_1 \oplus x_2=\mathsf{TRUE}$ and $x_1 \vee x_2=\mathsf{FALSE}$ correspond to the relations $x_1 x_2=-1$ and $x_1 + x_2 + x_1x_2=1$ respectively. Surprisingly, one can capture a wide variety of algebraic relations using only boolean constraints. However, capturing arbitrary relations using boolean constraints is not straightforward. One issue is that any variables which occur in the same constraint must commute in the BCS algebra. For example, if a BCS algebra includes the defining relation $x_1 + x_2 + x_1x_2=1$ then it also includes the defining relation $x_1x_2-x_2x_1=0$. However, we note that this does not preclude noncommutative relations from holding in BCS algebras. For example, the anticommutator $xy+yx=0$ is well-known to hold in the BCS algebra of the Mermin-Peres Magic Square \cite{Mer90,Per90,PS23}. Now, returning to the issue at hand, we remark that while several of the relations $r_1,\ldots,r_k$ are representable using boolean constraints, a number of them are not.

To handle these remaining relations, we prove that they belong to a more general family of relations, which we refer to as nested conjugacy BCS relations. Then, we derive a method for rewriting an arbitrary nested conjugacy BCS relation as a collection of BCS relations. The result is a computable embedding of any finitely presented $*$-algebra, consisting of nested conjugacy BCS relations, into a BCS algebra. Importantly, our embedding theorem is quantitative, which enables us to control the number of relations involved in the rewriting. This property is crucial for relating the positivity of $\alpha_M$ to the commuting operator value of the resulting BCS nonlocal game corresponding to the BCS algebra in which $r_1,\ldots,r_k$ hold. A corollary of our construction is a computable mapping from Turing machines $M$ to BCS nonlocal games $\mcG_M$, where the commuting operator value of $\mcG_M$ is strictly greater than $1/2$ if and only if $M$ halts.

\section*{Acknowledgments}
This project was initiated at the Banff International Research Station for Mathematical Innovation and Discovery (BIRS) workshop on Operator Systems and their Applications. The authors would like to thank both the conference organizers and BIRS for a productive workshop. M.F.~was supported by the European Research Council (ERC) under Agreement 818761 and by VILLUM FONDEN via the QMATH Centre of Excellence (Grant No. 10059). L.K.~acknowledges funding support from the Natural Sciences and Engineering Research Council of Canada (NSERC). A.M.~and C.P.~acknowledge the support of the Natural Sciences and Engineering Research Council of Canada (NSERC)(ALLRP-578455-2022). A.M.~is also supported by NSERC DG 2024-06049. D.R.~acknowledges the support of the Air Force Office of Scientific Research under award number FA9550-20-1-0375. W.S.~was supported by an NSERC Discovery Grant and an Alfred P. Sloan Research Fellowship. Y.Z.~is supported by VILLUM FONDEN via QMATH Centre of Excellence grant number 10059 and Villum Young Investigator grant number 37532. 

\section{Preliminaries}

\subsection{Computability theory}

We model \textbf{decision problems} as subsets $L \subseteq \Sigma^*$, where $\Sigma^*$ denotes the set of all finite strings over a fixed alphabet $\Sigma$. An input $x \in \Sigma^*$ is called a \textbf{yes} instance if $x \in L$ and a \textbf{no} instance otherwise. \textbf{Optimization problems} are specified by a set of valid instances $I \subseteq \Sigma^*$ together with an objective function $f: I \to \mathbb{R}$. While optimization problems ask for the maximum or minimum value of $f$, decision problems are concerned only with a yes/no answer. Given a threshold $\theta \in \mathbb{R}$, any optimization problem induces an associated threshold decision problem: given $x \in I$, decide whether $f(x) \geq \theta$ (or $f(x) >\theta$). Much of the complexity-theoretic analysis of optimization problems proceeds by studying the complexity of their associated threshold decision problems. A decision problem is called \textbf{decidable} if there exists some algorithm (\textit{i.e.,} a Turing machine) that returns the correct yes/no answer on every input, and undecidable otherwise. A famous example of an undecidable problem is the \textbf{Halting problem}, which asks whether a given algorithm halts on a given input.

Given two decision problems $A \subseteq \Sigma_A^*$ and $B \subseteq \Sigma_B^*$, a \textbf{many-one reduction} from $A$ to $B$ is a computable (\textit{e.g.,} by a Turing machine) map $f: \Sigma_A^* \to \Sigma_B^*$ such that $x \in A \Leftrightarrow f(x) \in B$. These reductions are used to compare the relative difficulty of problems.

We use \textbf{RE} to denote the class of decision problems for which there exists an algorithm that halts and accepts on every yes instance, but may runs indefinitely on no instances. The class \textbf{coRE} consists of the complement of problems in $\RE$; equivalently, problems in $\coRE$ admit an algorithm that halts and rejects on every no instance, but may run indefinitely on yes instances. In particular, the decidable problems are those that lie in the intersection of $\RE$ and $\coRE$. A problem $L$ is called \textbf{RE-hard} if every problem in $\RE$ can be reduced to $L$, via a many-one reduction, and \textbf{RE-complete} if it is both $\RE$-hard and in $\RE$. The analogous notions apply for the class $\coRE$. 

We use also use the notion of an \textbf{RE set}, $\mathcal{L} \subseteq \mathbb{N}$, which is a subset of the natural numbers such that there exists an algorithm that halts and accepts exactly on its elements. Since natural numbers can be encoded as binary strings, this provides a convenient way to treat languages as sets of numbers. A set $\mathcal{L} \subseteq \mathbb{N}$ is said to be RE-hard if for every RE set $\mathcal{L}' \subseteq \mathbb{N}$, there exists a computable reduction from $\mathcal{L}'$ to $\mathcal{L}$. In other words, deciding membership in $\mathcal{L}$ is at least as hard as deciding membership in any recursively enumerable set.  $\mathcal{L}$ is $\RE$-complete if it is both an $\RE$ set and $\RE$-hard. Such sets exist: for example, the halting problem can be encoded as a set of natural numbers by enumerating each Turing machine and input pair, and this set is the standard example of an $\RE$-complete set. Another key example of a decision problem is whether a given semialgebraic subset of $\mathbb{R}^n$ is empty. Formally:
\begin{problem} \label{prob:semialgebraic_set_problem}
    Given polynomials $p_1, \ldots, p_n \in \mathbb{R}[x_1, \ldots, x_k]$ and $g_1, \ldots, g_m \in \mathbb{R}[x_1, \ldots, x_k]$, decide whether the semialgebraic set
    \begin{equation*}
        \big\{ x \in \mathbb{R}^k \: : \: p_i(x) = 0, \; g_j(x) > 0 \big\},
    \end{equation*}
    is empty.
\end{problem}
The Tarski–Seidenberg decision method \cite{TA98,SA54,BS14}, based on quantifier elimination in the first order theory of reals, provides an exact algorithm for solving \Cref{prob:semialgebraic_set_problem}. In general, deciding whether a semialgebraic set is empty is not known to admit an efficient solution.

\subsection{Nonlocal games}

A two-player \textbf{nonlocal game} $\mathcal{G}$ is described by a tuple $(X, Y, A, B, \pi, V)$, where $X$ and $Y$ are finite sets of questions for two players (Alice and Bob), $A$ and $B$ are finite sets of possible answers, $\pi: X \times Y \to [0,1] $ is a probability distribution over $X \times Y$ specifying how questions are sampled, and $V: A \times B \times X \times Y \to \{0,1\}$ is a predicate indicating whether a pair of answers is accepted for a given pair of questions. The goal of the players is to maximize the probability of winning, that is, producing answers $(a,b) \in A \times B$ such that $V(a,b,x,y) = 1$ when questions $(x,y) \in X \times Y$ are sampled according to $\pi$. The players are not allowed to communicate during the game, but they can agree on a strategy beforehand.

Given a nonlocal game $\mcG$, the probability the players win employing a strategy $\mcS$ is given by:
\begin{equation*}
    \omega \big( \mathcal{G};\mcS \big) = \sum_{\substack{x \in X \\ y \in Y}} \pi(x,y) \sum_{\substack{a \in A \\ b \in B}} V(a,b,x,y) \Pr(a,b|x,y),
\end{equation*}
where $\Pr(a,b|x,y)$ is the probability that the players output $(a,b)$ given questions $(x,y)$ using strategy $\mcS$.

The \textbf{value} of a nonlocal game $\mathcal{G}$ for a class of strategies $\mcC$ is the supremum over all strategies from that class: $\omega_{\mcC}(\mcG)=\sup_{\mcS\in\, \mathcal{C}} \omega(\mcG;\mcS).$
In this paper, we focus on three subclasses of strategies: classical, quantum, and commuting operator, which yield the \textbf{classical value} $\omega_c(\mathcal{G})$, \textbf{quantum value} $\omega_q(\mathcal{G})$, and \textbf{commuting operator value} $\omega_{qc}(\mathcal{G})$, respectively.

\begin{definition}[Classical strategy] \label{def:classical_strategy}
    A \emph{classical strategy} for a two-player nonlocal game $\mathcal{G}$ consists of a shared randomness variable $\lambda$ taking values in a probability space $(\Lambda, \mathcal{F}, \mu)$ representing shared randomness and conditional probability distributions $\Pr\nolimits_A(a|x,\lambda)$ for Alice and $\Pr\nolimits_B(b|y,\lambda)$ for Bob. The resulting joint probability of outputs $(a,b)$ given inputs $(x,y)$ is
    \begin{equation*}
        \Pr(a,b|x,y) = \int_{\Lambda} \Pr\nolimits_A(a|x,\lambda) \cdot \Pr\nolimits_B(b|y,\lambda) \, d\mu(\lambda).
    \end{equation*}
\end{definition}

\begin{definition}[Quantum strategy] \label{def:quantum_strategy}
    A \emph{quantum strategy} for a two-player nonlocal game $\mathcal{G}$ consists of a shared quantum state $\ket{\psi}$ on a finite dimensional Hilbert space $\mathcal{H}_A \otimes \mathcal{H}_B$ and local PVMs $\{E^x_a\}_a$ on $\mathcal{H}_A$ and $\{F^y_b\}_b$ on $\mathcal{H}_B$. The joint probability of outputs $(a,b)$ given inputs $(x,y)$ is
    \begin{equation*}
        \Pr(a,b|x,y) = \bra{\psi} (E^x_a \otimes F^y_b) \ket{\psi}.
    \end{equation*}
\end{definition}

\begin{definition}[Commuting operator strategy] \label{def:commuting_operator_strategy}
    A \emph{commuting operator strategy} for a two-player nonlocal game $\mathcal{G}$ consists of a shared (possibly infinite-dimensional) Hilbert space $\mathcal{H}$, a quantum state $\ket{\psi} \in \mathcal{H}$, and families of PVMs $\{E^x_a\}_a$ and $\{F^y_b\}_b$ acting on $\mathcal{H}$ such that $[E^x_a, F^y_b] = 0$ for all $x,y,a,b$. The joint probability of outputs $(a,b)$ given inputs $(x,y)$ is
    \begin{equation*}
        \Pr(a,b|x,y) = \langle \psi | E^x_a F^y_b | \psi \rangle.
    \end{equation*}
\end{definition}

These classes of strategies satisfy the inclusions $\mathcal{S}_c \subseteq \mathcal{S}_q \subseteq \mathcal{S}_{qc}$, and hence the corresponding game values satisfy the relations
\begin{equation*}
    \omega_c(\mathcal{G}) \leq \omega_q(\mathcal{G}) \leq \omega_{qc}(\mathcal{G}).
\end{equation*}
Classical strategies are strictly weaker than quantum ones for many nonlocal games, for example the CHSH game \cite{Clauser1969} where $\omega_c(\mathcal{G}) < \omega_q(\mathcal{G})$. While $\mathcal{S}_q$ and $\mathcal{S}_{qc}$ coincide in finite dimensions, they differ in infinite dimensions, as shown in \cite{JNVWY21}.

When both players have the same input set, a \textbf{synchronous strategy} is defined by the condition
\begin{equation*}
    \Pr(a,b|x,x) = 0, \quad \forall a \neq b.
\end{equation*}
The supremum of the value over synchronous commuting-operator strategies is the \textbf{synchronous commuting-operator value}, and is denoted by $\omega^{\mathrm{sync}}_{qc}(\mathcal{G})$. A game is called \textbf{synchronous} if its predicate satisfies
\begin{equation*}
    V(a,b,x,x) = 0, \quad \forall a \neq b.
\end{equation*}
It is known that a synchronous game has a perfect strategy (i.e.~with value equal to $1$) if and only if it admits a perfect synchronous strategy. Nevertheless, there exist synchronous games whose optimal value is not attained by synchronous strategies \cite{Helton21}.

\subsection{The NPA hierarchy}\label{subsec:prelim-NPA}

The Navascués–Pironio–Acín (NPA) hierarchy \cite{NPA07,NPA08} is a hierarchy of semi-definite programming relaxations that yields outer approximations to the commuting operator value $\omega_{qc}(\mathcal{G})$ of a nonlocal game $\mathcal{G}$. Each level of the hierarchy reduces the feasible set, producing a nested sequence of bounds
\begin{equation*}
    \omega_{qc}(\mathcal{G}) \leq \cdots \leq \omega^{(2)}_{\mathrm{npa}}(\mathcal{G}) \leq \omega^{(1)}_{\mathrm{npa}}(\mathcal{G}),
\end{equation*}
where $\omega^{(k)}_{\mathrm{npa}}(\mathcal{G})$ denotes the value of the $k$-th level. The convergence theorem \cite{NPA08} guarantees that the sequence $\omega^{(k)}_{\mathrm{npa}}(\mathcal{G})$ converges to $\omega_{qc}(\mathcal{G})$, i.e.
\begin{equation*}
    \lim_{k \to \infty} \omega^{(k)}_{\mathrm{npa}}(\mathcal{G}) = \omega_{qc}(\mathcal{G}).
\end{equation*}
Hence, given a nonlocal game $\mathcal{G}$, for every $\varepsilon > 0$, there exists a $k\in \N$ such that
\begin{equation*}
    \omega^{(k)}_{\mathrm{npa}}(\mathcal{G}) < \omega_{qc}(\mathcal{G}) + \varepsilon.
\end{equation*}
Prior to this work it was not known if there exists a $k \in \mathbb{N}$ such that,
\begin{equation*}
    \omega^{(k)}_{\mathrm{npa}}(\mathcal{G}) = \omega_{qc}(\mathcal{G}).
\end{equation*}
Each relaxation $\omega^{(k)}_{\mathrm{npa}}(\mathcal{G})$ is expressed as the optimal value of an SDP whose feasible set is defined by the positivity of a moment matrix indexed by monomials of length at most $k$ in the measurement operators, subject to linear constraints. The exact definition of the moment matrix and the constraints can be found in \cite{NPA08}.

From a computational perspective, SDPs can be solved by approximate iterative algorithms that converge to the optimum under mild assumptions.\footnote{Numerical convergence requires that the primal and dual optima are attained. A sufficient condition is strict feasibility and boundedness of the feasible region, or strict feasibility of both primal and dual \cite{Tavakoli2024}.} Even in the absence of such conditions the decision problem
\begin{equation*}
    \text{Given a nonlocal game $\mathcal{G}$, decide if $\omega^{(k)}_{\mathrm{npa}}(\mathcal{G}) < \theta$}
\end{equation*}
can be phrased as a feasibility problem over a semialgebraic set, thereby falling under the scope of \Cref{prob:semialgebraic_set_problem}. Consequently, by the Tarski–Seidenberg decision method, there exists an exact algorithm for this decision problem \cite{Parrilo2003}.

\subsection{Algebras, states and groups}\label{sec:algebra}

In this paper, a \textbf{$*$-algebra } refers to a unital associative complex algebra $\mcA$ that is equipped with an antilinear involution ${}^*\colon \mcA \to \mcA$ such that $(ab)^* = b^*a^*$ for all $a,b \in \mcA$. For a $*$-algebra $\mcA$, we denote by $\mcA^\mathrm{op}$ the \textbf{opposite algebra}, i.e., the $*$-algebra with multiplication $a \cdot^\mathrm{op} b = ba$ for all $a,b\in \mcA^\mathrm{op}$ and other structure inherited by $\mcA$.

Given a set $\mcX$, we denote by $\C^*\langle \mcX\rangle$ the free non-commutative $*$-algebra generated by $\mcX$. Elements in  $\C^*\langle \mcX\rangle$ are called $*$-polynomials over $\mcX$. A set $\mcX$ is said to be a generating set for a $*$-algebra $\mcA$ if every element in $\mcA$ can be expressed by a $*$-polynomial over $\mcX$. A $*$-algebra is \textbf{finitely generated} if it has a finite generating set.
For a set of $*$-polynomials  $\mcR \subseteq \C^*\langle \mcX \rangle$, the quotient $\C^*\langle \mcX \rangle / \langle \mcR \rangle$  of the $*$-algebra by the two-sided $*$-ideal $\langle \mcR\rangle$ generated by $\mcR$ will be denoted by $\C^* \langle \mcX \colon \mcR \rangle$, and we often refer to the elements in $\mcR$ as relations. A $*$-algebra $\mcA$ is \textbf{finitely presented} if $\mcA=\C^* \langle \mcX \colon \mcR \rangle$ for some finite sets $\mcX, \mcR$. In a presentation $\C^* \langle \mcX \colon \mcR \rangle$, we sometimes also write $r=0$ for the relation $r\in\mcR$, and similarly write $a=b$ for the relation $a-b\in \mcR$. Another $*$-algebra we will use often is given by the \textbf{group algebra}: Given a group $G$, we define the group algebra $\C G := \mathrm{span}\{a\, g \, \vert \, a \in \C, g\in G\}$, where $(a\,g)\cdot(b\,h) = ab\, gh $, and $g^* = {g^{-1}}$ for $a,b\in \C$, $g\in G$. Similarly, a group presentation $G=\ang{\mcX:\mcR}$ denotes the quotient of $\mcF(\mcX)$, the free group on the set $\mcX$, by the normal subgroup generated by $\mcR\subseteq \mcF(\mcX)$. The distinction between the normal subgroup and the two-sided $*$-ideal means that group relations are typically written as ``$r=1$'', as opposed to ``$r=0$'', where $r$ is an element of $\mcF(\mcX)$ rather than $\C^*\ang{\mcX}$. Nonetheless, if $\mcR\subseteq \mcF(\mcX)$ is a set of group relations for the group $G=\ang{\mcX:\mcR}$ then $\C G=\C^*\langle \mcX:\widetilde{\mcR}\rangle$ is the $*$-algebra with (algebra) relations $\widetilde{\mcR}=\{1-xx^*,1-x^*x,1-r:x\in\mcX, r\in \mcR\}$.

Given two $*$-algebras $\mcA$ and $ \mcB$, a linear map $\pi \colon \mcA \to \mcB$ is a \textbf{$*$-ho\-mo\-mor\-phism} if it is multiplicative and $\pi(a^*) = \pi(a)^*$ for all $a\in \mcA$. We call a $*$-homomorphism an \textbf{embedding} if it is injective. Given three $*$-algebras $\mcA, \mcB, \mcC$ and an embedding $\kappa \colon \mcA \to \mcC$, we say that a $*$-homomorphism $\pi \colon \mcA \to \mcB$ \textbf{extends} to a ${}^*$-homomorphism $\hat\pi \colon \mcA \to \mcC$ if $\hat\pi \circ \kappa = \pi$. Similarly, if $\mcI$ is a two-sided $*$-ideal in $\mcA$, we say that a $*$-homomorphism $\pi \colon \mcA \to \mcB$ \textbf{descends} to a $*$-homomorphism $\tilde{\pi}\colon \mcA / \mcI \to \mcB$ if $\tilde{\pi}[x] := \pi(x)$ is well-defined. More generally, given two-sided $*$-ideals $\mcI \vartriangleleft \mcA$ and $\mcJ \vartriangleleft \mcB$ and a $*$-homomorphism $\pi \colon \mcA \to \mcB$, $\pi$ descends to a $*$-homomorphism $\Tilde{\pi} \colon \mcA/\mcI \to \mcB/\mcJ$ if $\Tilde{\pi}([x]_{A/\mcI}) = [\pi(x)]_{\mcB/\mcJ}$ is well-defined. In other words, the following diagram commutes:
\begin{equation*}
    \begin{tikzcd}
        \mcA \arrow{r}{\pi} \arrow[two heads]{d}{q_\mcI} & \mcB \arrow[two heads]{d}{q_\mcJ} \\
        \mcA/\mcI \arrow{r}{\Tilde{\pi}} & \mcB/\mcJ
    \end{tikzcd}
\end{equation*}
Given a triple consisting of a finitely presented $*$-algebra $\mcA=\C^*\langle \mcX:\mcR\rangle$, a $*$-algebra $\mcB$, and a $*$-homomorphism $\psi:\mcA\to \mcB$, the \textbf{lift} of $\psi$ is the unique $*$-homomorphism $\wtd{\psi}:\C^*\langle \mcX\rangle \to \mcB$ such that $\wtd{\psi}(r)=0$ for all $r\in R$. If $q_\mcR:\C^*\langle \mcX\rangle \to \mcA$ is the quotient map induced by the two-sided $*$-ideal generated by $\mcR$, then the above is summarized by the following commutative diagram.
\begin{equation}
\begin{tikzcd}
\C^*\langle \mcX\rangle \arrow{r}{\wtd{\psi}} \arrow[two heads]{d}{q_\mcR} &  \mcB\\
\mcA \arrow{ur}{\psi} &
\end{tikzcd}
\end{equation}
In particular, $\wtd{\psi}$ descends to $\psi$.

A linear functional $\varphi \colon \mcA \to \C$ from a $*$-algebra $\mcA$ is called a \textbf{state} if $\varphi$ is unital, $\varphi(a^*a) \geq 0$, and $\varphi(a^*) = \overline{\varphi(a)}$ for all $a\in \mcA$. All states considered in this paper are assumed to be \textbf{bounded}, that is, they satisfy  
\begin{align}\label{bounded}
    \sup\left\{\frac{\varphi(b^* a^* a b)}{\varphi(b^* b)}: b \in \mcA, \varphi(b^* b) \neq 0\right\}<\infty \text{ for all } a\in\mcA.
\end{align}
We say the state $\varphi$ is \textbf{tracial}, if $\varphi(ab) = \varphi(ba)$ for all $a,b\in \mcA$. Due to positivity, a state $\varphi$ on $\mcA$ gives rise to a seminorm $\Vert \cdot \Vert_\varphi \colon \mcA \to [0,\infty)$ by setting $\Vert a \Vert_\varphi = \sqrt{\varphi(a^*a)}$. The operator norm $\norm{\cdot}_\mcA:\mcA\arr\R_{\geq}\cup\{\infty\}$ is defined by $\norm{a}_\mcA=\sup\{\norm{a}_\varphi:\varphi \text{ a state on }\mcA \}$.

For a Hilbert space $\mcH$, we denote by $\mcB(\mcH)$ the $*$-algebra of bounded operators on the Hilbert space. Using this we can define a \textbf{$*$-representation} of a $*$-algebra $\mcA$ as a $*$-homomorphism $\pi \colon \mcA \to \mcB(\mcH)$. Given a $*$-representation $\pi \colon \mcA \to \mcB(\mcH)$ and a unit vector $\vert \xi \rangle \in \mcH$, we can define a state on $\mcA$ by setting $\varphi(a) = \langle \xi \vert \pi(a) \vert \xi \rangle$ for all $a\in \mcA$. Conversely, given a state $\varphi$ on $\mcA$, since $\varphi$ satisfies \Cref{bounded},
we can find a $*$-representation $\pi_\varphi \colon \mcA \to \mcB(\mcH_\varphi)$ on a Hilbert space $\mcH_\varphi$ and a unit vector $\vert \xi_\varphi \rangle\in \mcH_\varphi$ such that $\varphi(a) = \langle \xi_\varphi \vert \pi_\varphi(a) \vert \xi_\varphi \rangle$ by the GNS-construction \cite[Theorem 4.38]{Sch20}. It follows that $\norm{a}_\mcA=\sup\{\norm{\pi(a)}_{\mcB(\mcH)}:\pi:\mcA\arr\mcB(\mcH) \text{ a }*\text{-representation}\}$ for all $a\in \mcA$.

\subsection{Approximate representations}\label{sec:eps_rep}
We recall some terminology of approximate states from \cite{MSZ23}, as it will be used in our analysis of near-optimal strategies for nonlocal games.
\begin{definition}[Definition 3.1 in \cite{MSZ23}]
    Let $\mcA$ be a finitely generated $*$-algebra. Suppose $\eps \geq 0$ and 
    $\mcR\subseteq \mcA$. An \emph{$(\epsilon,\mcR)$-state on $\mcA$} is a state
    $\varphi$ on $\mcA$ such that $\varphi(r^*r)\leq \epsilon$ for all
    $r\in\mcR \cup \mcR^*$.
\end{definition}

In the above definition, $\varphi$ is a state on $\mcA$ that approximately respects all the relations in $\mcR$. Consequently, for any $\alpha$ in the two-sided $*$-ideal $ \ang{\mcR}$, one would expect $\varphi(\alpha)$ to be small as well. The notion of $\mcR$-decomposition, together with its size, provides a quantitative framework for this.

\begin{definition}[Definition 3.2 in \cite{MSZ23}]\label{defn:rde}
    Let $\mcA$ be a $*$-algebra with generating set
    $\mcX$. Let $\mcR\subseteq\C^*\ang{\mcX}$ be a set of $*$-polynomials over
    $\mcX$. For any $*$-polynomial $\alpha\in\C^*\ang{\mcX}$ that is trivial in
    $\mcA/\ang{\mcR}$, we say that $\sum_{i=1}^n \lambda_iu_ir_iv_i$ is an
    $\mcR$-decomposition for $\alpha$ in $\mcA$ if
    \begin{enumerate}
        \item $u_i,v_i$ are $*$-monomials in $\C^*\ang{\mcX}$ for all $1\leq i\leq n$, 
        \item $r_i\in\mcR\cup\mcR^*$ for all $1\leq i\leq n$,
        \item $\lambda_i \in \C$ for all $1\leq i\leq n$, and
        \item $\alpha=\sum_{i=1}^n \lambda_iu_ir_iv_i$ in $\mcA$.
    \end{enumerate}
    The size of an $\mcR$-decomposition $\sum_{i=1}^n \lambda_i u_i r_i v_i$ is $\sum_{i=1}^{n} |\lambda_i| (1+ \norm{r_i}_{\mcA} \deg(v_i))$,
    where $\norm{\cdot}_{\mcA}$ is the operator norm in $\mcA$.
 \end{definition}

As discussed in \Cref{sec:algebra}, we can regard elements of $\C^*\ang{\mcX}$ as elements of $\mcA$ via the natural $*$-homomorphism from $\C^*\ang{\mcX}\arr\mcA$. A $*$-monomial in $\C^*\ang{\mcX}$ of degree $k\geq 0$ is of the form $a_1a_2\ldots a_k$ where $a_1,\ldots,a_k\in\mcX\cup\mcX^*$.
 
It will become clear below why we need to keep track of the operator norms of relations in $\mcR$ and the degree of the monomials $v_i$'s (see \Cref{lemma:Rbound} for motivation and \Cref{prop:near-perfect} for a concrete example). The following lemma is useful for computing sizes of $\mcR$-decompositions. 

\begin{lemma}\label{lemma:sizecalc}
    Let $\mcA$ be a $*$-algebra with generating set $\mcX$, and let $\mcR
    \subseteq \C^*\ang{\mcX}$ be a set of $*$-polynomials over $\mcX$.
    \begin{enumerate}
        \item If $f_1,\ldots,f_n\in\C^*\ang{\mcX}$ have $\mcR$-decompositions in $\mcA$ of sizes $\Lambda_1,\ldots,\Lambda_n$, respectively, then for all $\lambda_1,\ldots,\lambda_n\in\C$, the $*$-polynomial $\sum_{i=1}^n\lambda_if_i$ has an $\mcR$-decomposition in $\mcA$ of size $\leq \sum_{i=1}^n\abs{\lambda_i}\Lambda_i$.
        \item If $\alpha_1,\alpha_2,\beta_1,\beta_2$ are $*$-monomials in $\C^*\ang{\mcX}$ such that $\alpha_1-\beta_1$ and $\alpha_2-\beta_2$ have $\mcR$-decompositions in $\mcA$ of sizes $\Lambda_1$ and $\Lambda_2$, respectively, then $\alpha_1\alpha_2-\beta_1\beta_2$ has an $\mcR$-decomposition in $\mcA$ of size $\leq \Lambda_2+\Lambda_1\left(1+\deg(\beta_2)\right)$.
    \end{enumerate}
\end{lemma}

\begin{proof}
    Part (1) follows straightforwardly from the definition of the size of an $\mcR$-decomposition.

    For part (2), let $\sum_i\lambda_i u_ir_iv_i$ and $\sum_j\wtd{\lambda}_j\wtd{u}_j\wtd{r}_j\wtd{v}_j$ be $\mcR$-decompositions of $\alpha_1-\beta_1$ and $\alpha_2-\beta_2$ in $\mcA$ of sizes $\Lambda_1$ and $\Lambda_2$. Then we have
    \begin{align*}
        \alpha_1\alpha_2-\beta_1\beta_2:=\alpha_1(\alpha_2-\beta_2)+(\alpha_1-\beta_1)\beta_2\\
        =\sum_j\wtd{\lambda}_j(\alpha_1\wtd{u}_j)\wtd{r}_j\wtd{v}_j+\sum_i\lambda_i u_ir_i(v_i\beta_2),
    \end{align*}
    which gives an $\mcR$-decomposition of $\alpha_1\alpha_2-\beta_1\beta_2$ in $\mcA$ of size 
    \begin{align*}
        \Lambda_2+\sum_i\abs{\lambda_i}\Big(1+\norm{r_i}_{\mcA}\big(\deg(v_i)+\deg(\beta_2)\big)\Big)&\leq \Lambda_2+\Lambda_1+\sum_i\abs{\lambda_i}\norm{r_i}_{\mcA}\deg(\beta_2)\\
        &\leq \Lambda_2+\Lambda_1+\Lambda_1\cdot\deg(\beta_2).
    \end{align*}
    The last inequality uses that $\sum_i\abs{\lambda_i}\norm{r_i}_{\mcA}\leq \Lambda_1$.
\end{proof}

 \begin{definition}[Definition 3.5 in \cite{MSZ23}]\label{defn:syncstate}
    Suppose $\mcA$ is a $*$-algebra.
    A state $\varphi$ on $\mcA \otimes \mcA$ is $(\eps,\mcX)$-synchronous for
    some $\eps \geq 0$ and $\mcX \subseteq \mcA$ if
    \begin{equation}
\varphi\big((x \otimes 1 - 1 \otimes x)^*(x \otimes 1 - 1 \otimes x)    \big)\leq \epsilon \label{eq:x11x}
    \end{equation}
    for all $x \in \mcX$.
\end{definition}

Recall that any state $\varphi$ on a $*$-algebra $\mcA$ induces a semi-norm $\norm{ \cdot }_\varphi$. Thus \Cref{eq:x11x} can be equivalently written as $\norm{x \otimes 1 - 1 \otimes x}_\varphi\leq \sqrt{\epsilon}$.

\begin{lemma}[Lemma 3.9 in \cite{MSZ23}]\label{lemma:Rbound}
    Let $\mcA$ be a $*$-algebra generated by a finite set of unitaries $\mcX$, and 
    let $\iota : \mcA \to \mcA \otimes \mcA : a \mapsto a \otimes 1$ be the
    left inclusion. Suppose $\mcR \subseteq \C^*\ang{\mcX}$ is a set of $*$-polynomials
    over $\mcX$, and let $\varphi$ be an $(\eps,\mcX)$-synchronous state on $\mcA \otimes \mcA$
    such that $\tau:=\varphi \circ \iota$ is an $(\eps,\mcR)$-state on $\mcA$. If $\alpha \in \mcA$ has an
    $\mcR$-decomposition of size $\Lambda$, then 
    $\norm{\alpha}_{\tau} \leq \Lambda \sqrt{\eps}$.
\end{lemma}

We refer the reader to \cite{MSZ23} for the proof. Alongside \cref{lemma:Rbound}, the following lemma will be useful for bounding the size of $\mcR$-decompositions under embeddings. 

\begin{lemma}\label{lemma:Rdecomp}
    Let $\mcA$ and $\wtd{\mcA}$ be $*$-algebras with generating sets $\mcX$ and $\wtd{\mcX}$, respectively, where $\mcX\subseteq \wtd{\mcX}$. Let $\mcR\subseteq \C^*\ang{\mcX}$ and $\wtd{\mcR}\subseteq\C^*\ang{\wtd{\mcX}}$ be sets of $*$-polynomials such that 
    \begin{enumerate}
        \item there are constants $C,\wtd{C}>0$ such that $\norm{r}_\mcA\geq C$ for all $r\in\mcR$ and $\norm{\wtd{r}}_{\wtd{\mcA}}\leq \wtd{C}$ for all $\wtd{r}\in \wtd{\mcR}$,
        \item the $*$-homomorphism from $\C^*\ang{\mcX}\arr\C^*\ang{\wtd{\mcX}}$ sending $x\mapsto x$ for all $x\in\mcX$ descends to a $*$-homomorphism from $\mcA/\ang{\mcR}\arr \wtd{\mcA}/\ang{\wtd{\mcR}}$, and
        \item there is a constant $\Delta>0$ such that every $r\in\mcR$ has an $\wtd{\mcR}$-decomposition in $\wtd{\mcA}$ of size $\leq \Delta$. 
    \end{enumerate}
Suppose $f\in\C^*\ang{\mcX}$ is trivial in $\mcA/\ang{\mcR}$ and has an $\mcR$-decomposition in $\mcA$ of size $\Lambda$, then $f$ has an $\wtd{\mcR}$-decomposition in $\wtd{\mcA}$ of size $\leq \big(1+\tfrac{\wtd{C}}{C}\big)\Delta\cdot \Lambda$.
\end{lemma}

Hypothesis (2) in the above lemma guarantees that every $r\in\mcR$ is trivial in $\wtd{\mcA}/\ang{\wtd{\mcR}}$, so it is meaningful to consider their $\wtd{\mcR}$-decompositions in $\wtd{\mcA}$. It is also important to note that the $*$-homomorphism $\mcA/\ang{\mcR}\arr \wtd{\mcA}/\ang{\wtd{\mcR}}$ in (2) need not be injective.

\begin{proof}
    Let $f=\sum_i\lambda_i u_i r_i v_i$ be an $\mcR$-decomposition in $\mcA$ of size $\Lambda$, where $\lambda_i\in\C$, $u_i,v_i$ are monomials in $\C^*\ang{\mcX}$, and $r_i\in \mcR\cup\mcR^*$. By hypothesis (3), every $r_i$ has an $\wtd{\mcR}$-decomposition $r_i=\sum_j \wtd{\lambda}^{(i)}_j \wtd{u}^{(i)}_j \wtd{r}^{(i)}_j\wtd{v}^{(i)}_j $ in $\wtd{\mcA}$ of size $\leq \Delta$, where $\wtd{\lambda}^{(i)}_j\in\C$, $\wtd{u}^{(i)}_j,\wtd{ v}^{(i)}_j$ are monomials in $\C^*\ang{\wtd{\mcX}}$, and $ \wtd{r}^{(i)}_j\in \wtd{\mcR}\cup\wtd{\mcR}^*$. So 
    \begin{equation*}
        \sum_i\sum_j \lambda_i\wtd{\lambda}^{(i)}_j \left(u_i\wtd{u}^{(i)}_j\right) \wtd{r}^{(i)}_j\left(\wtd{v}^{(i)}_j v_i\right)\label{eq:Rdecomp}
    \end{equation*}
    is an $\wtd{\mcR}$-decomposition of $f$ in $\wtd{\mcA}$.
     By hypothesis (1), $\norm{\wtd{r}^{(i)}_j}_{\wtd{\mcA}}\leq \tfrac{\wtd{C}}{C}\norm{r_i}_\mcA$ for all $i,j$. It follows that the size of the above $\wtd{\mcR}$-decomposition is
    \begin{align*}
       &\sum_i\sum_j \abs{\lambda_i\wtd{\lambda}^{(i)}_j}\left(1+\norm{\wtd{r}^{(i)}_j}_{\wtd{\mcA}} \deg\left(\wtd{v}^{(i)}_j v_i  \right)  \right)\\
       &\leq \sum_i \abs{\lambda_i}  \sum_j     \abs{\wtd{\lambda}^{(i)}_j}\left(1+\norm{\wtd{r}^{(i)}_j}_{\wtd{\mcA}} \big(\deg(\wtd{v}^{(i)}_j)+ \deg(v_i   )\big) \right)\\
       & = \sum_i \abs{\lambda_i}  \sum_j     \abs{\wtd{\lambda}^{(i)}_j}\left(1+\norm{\wtd{r}^{(i)}_j}_{\wtd{\mcA}} \deg(\wtd{v}^{(i)}_j)\right) + \sum_i \abs{\lambda_i}  \sum_j     \abs{\wtd{\lambda}^{(i)}_j}\cdot\norm{\wtd{r}^{(i)}_j}_{\wtd{\mcA}}\deg(v_i   )\\
       &\leq \sum_i \abs{\lambda_i} \Delta + \sum_i \abs{\lambda_i}  \left(\sum_j     \abs{\wtd{\lambda}^{(i)}_j}\right)\tfrac{\wtd{C}}{C}\norm{r_i}_{\mcA}\deg(v_i)\\
       &\leq \Lambda \cdot \Delta +\Delta \cdot \tfrac{\wtd{C}}{C} \sum_i \abs{\lambda_i} \norm{r_i}_{\mcA}\deg(v_i)\\
       &\leq \Lambda \cdot \Delta + \Delta \cdot \tfrac{\wtd{C}}{C}\cdot \Lambda =\big(1+\tfrac{\wtd{C}}{C}\big)\Delta\cdot \Lambda.
    \end{align*}
    The above inequalities use that both $\sum_i\abs{\lambda_i}$ and $\sum_i \abs{\lambda_i} \cdot\norm{r_i}_{\mcA}\deg(v_i)$ are $\leq \sum_i \abs{\lambda_i}\big(1+ \norm{r_i}_{\mcA}\deg(v_i)\big)=\Lambda$ and that $\sum_j\abs{\wtd{\lambda}^{(i)}_j}\leq \Delta$.
\end{proof}

\section{Boolean constraint system nonlocal games and BCS algebras}\label{sec:BCS}

We briefly recall some definitions and notation used for constraint systems. Our presentation here closely follows that of \cite{PS23}, but the concepts are largely due to \cite{AKS17,CLS17,CM14,J13}.

\subsection{Boolean constraint systems and nonlocal games}

We employ the multiplicative convention for boolean values, with $-1$ and $1$ representing $\mathrm{TRUE}$ and $\mathrm{FALSE}$ respectively. A \textbf{boolean relation} of arity $\ell>0$ is a subset $\msR$ of $(\pm1)^\ell$. Each relation can also be described in terms of its \textbf{indicator function} $f_\msR:\{\pm1\}^\ell \to \{\pm1\}$ with the property that $f_\msR^{-1}(-1)=\msR$ and $f_\msR^{-1}(1)=\{\pm1\}^\ell\setminus \msR$. Given a set of boolean variables $\mcX$, a \textbf{boolean constraint} $\msC$ on $\mcX=(x_1,\ldots,x_n)$, is a pair $\msC=(\msU,\msR)$, where $\msU\subseteq \mcX$ is the subset of constrained variables, called the \textbf{context}, and $\msR\subset \{\pm 1\}^{|\msU|}$ is an $\ell$-ary relation, with $\ell=|\msU|$. A \textbf{boolean assignment} of the variables $\mcX$ is a function $\phi:\mcX\to \{\pm1\}$. An assignment $\phi$ is a \textbf{satisfying assignment} for $\msC=(\msU,\msR)$ if $\{(\phi(x_{i_1}),\ldots, \phi(x_{i_{|\msU|}}))\}\in \msR$, where each $x_{i_j}\in \msU$ for $1\leq j \leq |\msU|$ and $i_1<i_2<\cdots<i_{|\msU|}$, otherwise we say that $\phi$ is an \textbf{unsatisfying assignment}. For convenience, we will write $\phi(\msU)$ to denote to the image of $\phi$ restricted to $\msU\subset \mcX$, in particular we write $\phi(\msU)\in \msR$ if and only if $\phi$ is a satisfying assignment for $\msC$. Furthermore, when $\mcX=\{x_1,\ldots, x_n\}$ we write $j\in \msU_i$ to index the variable $x_j$ appearing in $\msU_i$, and $\phi_j=(x_j)$ for $\phi\in \Z_2^{\mcX}$ an assignment restricted to variable with index $j$.

\begin{definition}\label{def:bcs}
A boolean constraint system (BCS) is a pair $(\mcX,\{\msC_i\}_{i=1}^k)$, where $\mcX$ is a set of boolean variables, and each $\{\msC_i\}_{i=1}^k$ a collection of constraints. Moreover, an assignment $\phi:\mcX\to \{\pm 1\}$ is a satisfying assignment for the constraint system if and only if $\phi(\msU_i)\in \msR_i$ for all $1\leq i \leq k$. That is the assignment satisfies all constraints simultaneously.
\end{definition}

Provided with a distribution over constraints $\mu:[k]\to \R_{\geq 0}$, each BCS $\msB$ gives rise to the following nonlocal game:

\begin{definition}\label{def:bcs_game}
A (two-player) BCS nonlocal game $\mcG(\msB,\mu)$ is a one-round interactive protocol between two spatially separated (i.e.~non-communicating) players (called provers), and a referee (called the verifier):
\begin{itemize}
    \item The first player (often named Alice) is given a constraint $\msC_i$ with probability $\mu(i)$, and the second player (often named Bob) is given a variable $x_j$ from that constraint chosen uniformly at random over $\msU_i$,
    \item Alice responds with an assignment $\phi$ of the variables $\msU_i$ in the constraint, while Bob responds with a $\pm1$-assignment $\gamma$ to their variable $x_j$, and
    \item they win the game if the assignment $\phi$ is satisfying for $\msC_i$, and both assignments agree on the distinguished variable, i.e.~ $\gamma=\phi(x_j)$. These conditions are checked by the referee.
\end{itemize}
\end{definition} 

We often omit the specification of the probability measure $\mu$ and simply write $\mcG(\msB)$ for $\mcG(\msB,\mu)$.

\begin{remark}
In \cite{PS23} the authors considered a different nonlocal game associated with a BCS. In this other variant, each player receives a constraint and replies with a full assignment to their respective context. The winning condition here is that their assignments must be satisfying for their respective constraints, and they must agree on any intersecting variables. The reason for considering this variant, is that the constraint-constraint version is essentially a \emph{synchronous} nonlocal game, and therefore its perfect strategies have a very nice algebraic characterization. Readers interested in synchronous games should consult the extensive literature on these games, see \cite{HMPS19,PSSTW16} and references within. Following \cite{CM24} we refer to these two types of games, as the \emph{constraint-constraint} and \emph{constraint-variable} games associated with $\msB$.
\end{remark}

From now on (unless stated otherwise) we consider the constraint-variable BCS game. In the setting of nonlocal games, the strategies employed by the players can be probabilistic. The probability of winning a BCS nonlocal games $\mcG(\msB)$ using strategy $\mcS$ is given by
\begin{equation}
    \omega(\mcG(\msB);\mcS)=\sum_{i\in [k],j\in \msU_i}\mu(i)/|\msU_i|\sum_{\substack{\phi\in \Z_2^{|\msU_i|},\gamma\in \Z_2\\:\phi(\msU_i)\in \msR_i \wedge \phi(x_j)=\gamma}} \Pr(\phi,\gamma|i,j),
\end{equation}
where $\Pr(\phi,\gamma|i,j)$ is the probability that the players respond with assignments $\phi\in \Z_2^{|\msU_i|}$ and $\gamma\in \Z_2$, given constraint $i$ and variable $j$. For $\epsilon\geq 0$, we say that $\mcS$ is \textbf{$\epsilon$-perfect} for $\mcG(\msB)$ if $\omega(\mcG;\mcS)\geq 1-\epsilon$; and we say that $\mcS$ is \textbf{perfect} if it is $\epsilon$-perfect with $\epsilon=0$. One of the many insights from the work \cite{CM14} was to consider the expected \emph{bias} ($p_{win}-p_{lose}$) on an input $(i,j)$ of a BCS game, as it related the $2$-outcome measurements with $\pm1$-valued observables.

\begin{definition}\label{def:co_strat}
   Let $\mcX=\{x_1,\ldots,x_n\}$ be a set of boolean variables, and let $\msB=(\mcX,\{\msC_i\}_{i=1}^k)$ be a boolean constraint system. A \emph{commuting operator strategy} 
    \begin{equation*}
        \mcS=\left( \ket{\psi}\in \mcH,\big\{\{P_\phi^{(i)}:\phi \in \Z_2^{|\msU_i|}\}:1\leq i \leq k\big\},\big\{\{Q_\gamma^{(j)}:\gamma \in \Z_2\}:1\leq j \leq n\big\} \right)
    \end{equation*}
for $\mcG(\msB)$ consists of
    \begin{enumerate}
        \item a Hilbert space $\mcH$ with a unit vector $|\psi\rangle\in \mcH$,
        \item collections of PVMs 
        \begin{equation*}
            \{P_\phi^{(i)}:\phi \in \Z_2^{|\msU_i|}\},1\leq i \leq k
        \end{equation*}
        and
        \begin{equation*}
            \{Q_\gamma^{(j)}:\gamma \in \Z_2\},1\leq j \leq n,
        \end{equation*}
such that $[P_\phi^{(i)},Q_\gamma^{(j)}]=0$ for all $\phi,\gamma,i,j$.
    \end{enumerate}
    When the indices are unambiguous, we simply write 
    \begin{equation*}
    \mcS=\left(\ket{\psi}\in \mcH,\{P^{(i)}_\phi\},\{Q^{(j)}_\gamma\}   \right); \quad \Pr(\phi,\gamma|i,j)=\langle \psi|P_\phi^{(i)}Q_{\gamma}^{(j)}|\psi\rangle.
\end{equation*}

\end{definition}

In this paper, we always assume that a strategy
\begin{equation*}
    \mcS=\left(\ket{\psi}\in \mcH,\{P^{(i)}_\phi\},\{Q^{(j)}_\gamma\}   \right)
\end{equation*}
employs projection-valued measures (PVMs). Hence every $X_j:=Q^{(j)}_1-Q^{(j)}_{-1}$ is a unitary of order $2$, and the projections $Q^{(j)}_\gamma=\tfrac{1+\gamma X_j}{2}$ for $\gamma\in\Z_2$ are the spectral projections onto the $\gamma$-eigenspace of $X_j$. We will often use the alternative notation 
\begin{equation*}
    \mcS=\left(\ket{\psi}\in \mcH,\{P^{(i)}_\phi\},\{X_j\}   \right)
\end{equation*}
to denote such a commuting operator strategy.

\subsection{BCS algebras, states, and commuting operator strategies}

We now move on to discuss an algebraic characterization of perfect strategies for commuting operator strategies. Given a boolean constraint system $\msB$ over boolean variables $\mcX=\{x_1,\ldots,x_n\}$, we also regard $\mcX$ as a set of noncommutative indeterminates, and consider $\C^*\ang{\mcX}$, the free $*$-algebra generated by $\mcX$, and the quotient 
\begin{equation*}
    \C\Z_2^{*\mcX}=\C^*\ang{\mcX:x^*x=xx^*=x^2=1 \text{ for all }x\in\mcX},
\end{equation*}
which encodes the relations of boolean variables as unitaries of order two, called the $*$-algebra generated by $\mcX$ over $\Z_2$. 

\begin{definition}
    Let $\mcX=\{x_1,\ldots,x_n\}$ be a set of boolean variables, and let $\msB=(\mcX,\{\msC_i\})$ be a boolean constraint system. For any commuting operator strategy $\mcS=\left(\ket{\psi}\in \mcH,\{P^{(i)}_\phi\},\{X_j\}   \right)$ for $\mcG(\msB)$, we define its \emph{associated representation} $\pi_\mcS$ to be the representation $\C\Z_2^{*\mcX}\arr\msB(\mcH)$ sending 
    \begin{equation*}
        x_j\mapsto X_j
    \end{equation*}
    for all $1\leq j\leq n$,
and we define the \emph{associated state} $\varphi_\mcS$ on $\C\Z_2^{*\mcX}$ by 
\begin{equation*}
    \varphi_{\mcS}(\alpha)=\bra{\psi}\pi_\mcS(\alpha)\ket{\psi}
\end{equation*}
for all $\alpha\in \C\Z_2^{*\mcX}$.
\end{definition}

\begin{definition}\label{def:bcs_alg}
 For any boolean constraint $\msC=(\msU,\msR)$, the associated BCS relations $\mcR(\msC)$ is a subset of $\C^*\ang{\msU}$ consisting of 
 \begin{enumerate}
    \item [(B0)] the commutation relations 
    \begin{equation*}
        xy-yx
    \end{equation*}
     for all $x,y\in\msU$, and
    \item[(B1)] the constraint relation 
    \begin{equation*}
        \sum_{\phi:\,\phi(\msU)\notin\msR} \prod_{x\in \msU}\tfrac{1}{2}(1+\phi(x)x).
    \end{equation*}
\end{enumerate}
For any boolean constraint system $\msB=\big(\mcX,\{\msC_i\}_{i=1}^{k}\big)$, the \emph{associated BCS relations} $\mcR(\msB)$ is defined as
\begin{equation*}
    \mcR(\msB):=\bigcup_{i=1}^k \mcR(\msC_i)\subseteq\C^*\ang{\mcX},
\end{equation*}
and \emph{the associated BCS algebra} $\mcA(\msB)$ is defined as the quotient
\begin{equation*}
    \mcA(\msB):=\C\Z_2^{*\mcX}/\ang{\mcR(\msB)},
\end{equation*}
where $\ang{\mcR(\msB)}$ denotes the two-sided $*$-ideal generated by relations in $\mcR(\msB)$. 
\end{definition}

We remark that relations (B0) ensure that the variables in any assignment are jointly measurable, while relations (B1) ensure that only projections onto satisfying assignments are present. We also remark that different sets $\mcR(\msC_i)$ and $\mcR(\msC_j)$ may contain the same commutation relation $xy-yx$ (this happens when $x,y\in\msU_i\cap\msU_j$). In this case, this commutation relation is included only once in $\mcR(\msB)$. Although not defined formally, the notion of the BCS algebra was implicit in \cite{J13,CM14}.

\begin{remark}
    Equivalently, one can express relation (B1) in terms of the indicator function for a boolean constraint. This way, (B1) is equivalent to the algebra relation:
    \begin{equation}\label{eq:multi_bcs_con}   
    2^{-|\msU|}\sum_\phi f_\msR(\phi(\msU))\prod_{x\in \msU}(1+\phi(x)x)=-1.
    \end{equation}
    Expanding the left hand side of \cref{eq:multi_bcs_con} gives us a polynomial $P_\msR(\msU) \in \C\Z_2^{\msU}$ such that $P_\msR(\msU)=-1$ if and only if (B1) holds for $\msC$ in \cref{def:bcs_alg}.
    \end{remark}

\begin{example}\label{example:BCS}
    For a boolean variable $x_i\in\mcX$, we let $p_i:=\frac{1-x_i}{2}$ be the $-1$-eigenspace projection of $x_i$ in $\C\Z_2^{*\mcX}$.    \Cref{table:BCS} summarizes constraint relations in $\mcR(\msB)$ associated with some common boolean constraints $\msB$. We present relations using both boolean variables $x_i$ and self-adjoint projections $p_i$. We denote the trivial relation involving $x_1$ and $x_2$ by $\mathcal{T}(x_1,x_2)$.

\begin{table}[ht]
\centering 
\renewcommand{\arraystretch}{1.4}
\begin{tabular}{|c|c|c|}
\hline
\textbf{Boolean constraint} & 
\textbf{Relation in $x_i$} & 
\textbf{Relation in $p_i$} \\
\hline
$x_1 = x_2$ & $x_1=x_2$ & $p_1 + p_2 - 2p_1 p_2=0$ \\
\hline
$x_1 \wedge x_2=\mathsf{TRUE}$ & $\tfrac{1}{2}(1 + x_1 + x_2 - x_1x_2)=-1$ & $1-p_1p_2=0$ \\
\hline
$x_1 \oplus x_2=\mathsf{TRUE}$ & $x_1 x_2=-1$ & $1-p_1 - p_2 + 2p_1 p_2=0$ \\
\hline
$x_1 \vee x_2=\mathsf{FALSE}$ & $x_1 + x_2 + x_1x_2=1$ & $p_1 + p_2 - p_1 p_2=0$ \\
\hline
$x_1=x_2\wedge x_3$ & $x_1=\tfrac{1}{2}(1 + x_2 + x_3 - x_2x_3)$ & $p_1-p_2p_3=0$ \\
\hline
$\mathcal{T}(x_1,x_2)$ & $\emptyset$ & $\emptyset$ \\
\hline
\end{tabular}
    \caption{Algebraic relations associated with some common boolean constraints. Note that we have ignored the commutation relations for simplicity.}
    \label{table:BCS}
\end{table}
Although the commutation relations have been ignored in \Cref{table:BCS} for simplicity, we do note that the trivial constraint relation still invokes the commutation relation $x_1x_2=x_2x_1$ in the BCS algebra.
\end{example}

 One of the main advantages to working with BCS nonlocal games is the following result, characterizing the existence of perfect commuting operator strategies in terms of the algebra.

\begin{proposition}{\cite[Theorem 3.11(4)]{PS23}}\label{prop:trace_perfect}
There is a tracial state on the BCS algebra $\mcA(\msB)$ if and only if $\mcG(\msB)$ has a perfect commuting operator strategy.
\end{proposition}

A particularly important example are BCS algebras consisting entirely of linear constraints. In this case the corresponding nonlocal games are called linear constraint system (LCS) nonlocal games. The theory of these nonlocal games is extremely rich, and contains many interesting examples, such as the famous Mermin-Peres Magic Square, among other important results \cite{SV18,S19,S20}.

\begin{example}{\cite[Example 3.9]{PS23}}\label{ex:lcs_alg}
Suppose $Ax=b$ is an $m\times n$ linear system over $\Z_2$, written in the conventional additive way. This system of \emph{linear} constraints gives rise to the BCS $\msB=(\mcX,\{\msU_i,\msC_i\}_{i=1}^m)$, where $\mcX=\{x_1,\ldots,x_n\}$, contexts $\msU_i=\{x_j:A_{ij}\neq 0\}$, and $\msC_i$ is the $i$th (row) equation of the linear system, written multiplicatively each constraint is $x_1^{A_{i1}}\cdots x_n^{A_{in}}=(-1)^{b_i}$. The BCS algebra $\mcA(\msB)$ is the finitely-presented $*$-algebra $\C \Z_2^{*\mcX}$ with additional relations:
\begin{enumerate}
    \item[(L0)] $x_1^{A_{i1}}\cdots x_n^{A_{in}}+(-1)^{b_i+1}$ for all $1\leq i \leq m$, and
    \item[(L1)] $x_jx_k-x_kx_j$ for all $x_j,x_k\in \msU_i$, $1\leq i \leq m$.
\end{enumerate}
\end{example}

The relations above are very close to those of a group algebra.

\begin{lemma}\label{lem:lcs_group}
Given an $m\times n$ linear system $Ax=b$ over $\Z_2$, let $\Gamma(A,b)$ be the finitely presented group with generating set $\mcX\cup \{J\}$ and (group) relations:
\begin{enumerate}
    \item[(G0)] $J^2$, and $x_i^2$ for all $1\leq i \leq n$,
    \item[(G1)] $x_iJx_i^{-1}J^{-1}$ for all $1\leq i \leq n$,
    \item[(G2)] $J^{b_i}x_1^{A_{i1}}\cdots x_n^{A_{in}}$ for all $1\leq i \leq m$, and
    \item[(G3)] $x_jx_kx_j^{-1}x_k^{-1}$ for all $x_j,x_k\in \msU_i$, $1\leq i \leq m$,
\end{enumerate}
then it holds that $\mcA(\msB)=\C\Gamma(A,b)/\langle J+1\rangle$, where $\msB$ is the BCS outlined in \cref{ex:lcs_alg}.  
\end{lemma}

The group $\Gamma(A,b)$ is called the \emph{solution group} associated with the LCS $Ax=b$, see \cite{CLS17} for more details. With \cref{ex:lcs_alg} in mind, we refer to that relations of the form $(L0)\cup(L1)$ as \emph{linear relations}, as any such relations can be enforced in a BCS algebra by including the appropriate linear constraints.

We now state a technical lemma concerning LCS nonlocal games and solution groups, which we will need later as part of main proof. The result essentially follows from properties of the main embedding theorem in \cite{S20} for groups, but is general enough to be restated here.

\begin{lemma}\label{lemma:LCSembedding}
Let $G=\ang{\mcX:\mcR}$ be a finitely presented group such that $x^2=1$ in $G$ for all $x\in \mcX$.
Then there exists a boolean constraint system $\msB=(\wtd{\mcX},\{\msC_i\})$ and a positive constant $C_G$ such that
\begin{enumerate}
    \item the natural $*$-homomorphism $\C^*\ang{\mcX}\arr\C^*\ang{\wtd{\mcX}}$ descends to an embedding $\C G\hookrightarrow \mcA(\msB)$,
    \item $\mcA(\msB)$ has a tracial state, 
    \item for every $r\in\mcR$, the $*$-polynomial $1-r\in\C^*\ang{\mcX}$ has an $\mcR(\msB)$-decomposition in $\C\Z_2^{*\wtd{\mcX}}$ of size $\leq C_G$, and
    \item the mapping $(\mcX,\mcR)\mapsto (\msB,C_G)$ is computable.
\end{enumerate}
\end{lemma}

\begin{proof} The proof essentially follows from the embedding theorem in \cite{S20}. In particular, \cite[Theorem 3.1]{S20} shows that given a finitely presented group $G$, there exists a system of linear constraints $Ax=b$ and an associated solution group $\Gamma(A,b)$, such that $G\times \Z_2\hookrightarrow \Gamma(A,b)$ via the natural embedding $\langle \mcX\rangle \to \langle \widetilde{\mcX}\rangle$ and $\langle z\rangle\to \langle J\rangle$, where $\wtd{\mcX}$ is the set of variables in the linear system $Ax=b$ and contains $\mcX$, and $z$ is the nontrivial element in $\Z_2$. Since $z\neq 1$ in $G\times \Z_2$, we have that $J\neq 1$ in $\Gamma(A,b)$. Moreover, by \Cref{lem:lcs_group}, there is a boolean constraint system $\msB=(\wtd{\mcX},\{\msC_i\})$ such that $\mcA(\msB)=\C\Gamma(A,b)/\ang{J+1}$. It follows that 
$\C G= \C G\times \Z_2/\ang{z+1}\hookrightarrow \C\Gamma(A,b)/\ang{J+1}=\mcA(\msB)$. This proves part (1). By \cite[Theorem 4]{CLS17}, $J\neq 1$ in $\Gamma(A,b)$ implies that the BCS game $\mcG(\msB)$ has perfect commuting operator strategies. So part (2) follows from \Cref{prop:trace_perfect}.

The embedding theorem in \cite{S20} is constructive, and therefore the mapping $(\mcX,\mcR)\mapsto \msB$ is computable. For every group relation $r\in\mcR$, one can search through the $\mcR(\msB)$-decompositions in $\C\Z_2^{*\wtd{\mcX}}$ of the $*$-algebra relation $1-r$. Since $1-r$ is trivial in $\mcA(\msB)$, this process will always halt and give a decomposition, furthermore, the size $C_r$ of the resulting $\mcR(\msB)$-decomposition is also computable. Since, there are finitely many relations in $\mcR$, $C_G:=\max_{r\in\mcR} C_r$ is computable. This establishes (3) and (4).
\end{proof}

Unlike several previous works on BCS algebras, the characterization of perfect strategies will not suffice here. Fortunately, the correspondence between certain representations of the BCS algebra and commuting operator strategies persists in the approximate case as well. However, a suitable notion of approximate tracial states is required. Such concepts are borrowed from approximate representation theory, as seen in \cref{sec:eps_rep}, and have been used to great effect to examine nonlocal games, see for example \cite{P22,Vid22,SV18,CM24,MS23,zhao24}.

\begin{proposition}\label{prop:near-perfect}
    Let $\mcX=\{x_1,\ldots,x_n\}$ be a set of boolean variables, and let $\msB=\big(\mcX,\{\msC_i=(\msU_i,\msR_i)\}_{i=1}^k\big)$ be a boolean constraint system. Let $M:=\max_{1\leq i\leq k}\{\abs{\msU_i}\}$ be the maximal size of contexts,  and let $T_{\msB}:=4^{M+2}kM^3$. Suppose $\mcS=\left(\ket{\psi}\in \mcH,\{P^{(i)}_\phi\},\{X_j\}   \right)$ is an $\epsilon$-perfect commuting operator strategy for the BCS game $\mcG(\msB)$ and let $\varphi_\mcS$ be the associated state on $\C\Z_2^{*\mcX}$. Then
    \begin{enumerate}
     \item there is  a $\left(T_{\msB}\cdot\epsilon,\mcX\right)$-synchronous state $f$ on $\C\Z_2^{*\mcX}\otimes \C\Z_2^{*\mcX}$ such that $\varphi_\mcS=f\circ \iota$, where $\iota:\C\Z_2^{*\mcX}\hookrightarrow \C\Z_2^{*\mcX}\otimes \C\Z_2^{*\mcX} $ is the left inclusion, and
        \item $\varphi_S$ is a $(T_\msB\cdot\epsilon,\mcR(\msB))$-state, where $\mcR(\msB)$ are the BCS relations associated with $\msB$.
    \end{enumerate}
\end{proposition}

A similar result to \cref{prop:near-perfect} was established in \cite[Proposition 4.9]{P22} for the case of finite-dimensional quantum strategies. We provide a proof of the full commuting operator generalization here.

\begin{proof}[Proof of \cref{prop:near-perfect}]
 For every $1\leq j\leq n$, we fix a constraint $\msC_i=(\msU_i,\msR_i)$ such that $x_j\in\msU_i$ and let 
 \begin{equation*}
     Y_j:=\sum_{\phi:\phi(\msU_i)\in\msR_i}\phi(x_j)P^{(i)}_{\phi}.
 \end{equation*}
Then $Y_j$ is a binary observable that commutes with all $X_1,\ldots,X_n$.  Since $\mcS$ has a winning probability $1-\epsilon$, and there are at most $kM$ pairs of questions, the winning probability for each pair of questions is at least $1-kM\epsilon$. Computing the bias for the question pair $(\mcC_i,x_j)$ yields
    \begin{equation*}
        \bra{\psi}Y_jX_j\ket{\psi}\geq 1- 2kM\epsilon.
    \end{equation*}
    Let $\pi:\C\Z_2^{*\mcX}\otimes \C\Z_2^{*\mcX}\arr \mcB(\mcH)$ be the representation sending $x_j\otimes 1\mapsto X_j$ and $1\otimes x_j\mapsto Y_j$ for all $1\leq j\leq n$, and let $f$ be the state on $\C\Z_2^{*\mcX}\otimes \C\Z_2^{*\mcX}$ defined by $f(\alpha)=\bra{\psi}\pi(\alpha)\ket{\psi}$. It follows that
\begin{align*}
    \norm{x_j\otimes 1-1\otimes x_j}_f^2&=\norm{X_j\ket{\psi}-Y_j\ket{\psi}}^2\\
    &=2-2\bra{\psi}Y_jX_j\ket{\psi}\\
    &\leq 4kM\epsilon\leq T_{\msB}\cdot\epsilon
\end{align*}
for all $1\leq j\leq n$. This means $f$ is a $\left(T_{\msB}\cdot\epsilon,\mcX\right)$-synchronous state. For any $\alpha\in \C\Z_2^{*\mcX}$, $f(\alpha\otimes 1)=\bra{\psi}\pi(\alpha\otimes 1)\ket{\psi}=\bra{\psi}\pi_\mcS(\alpha)\ket{\psi}=\varphi_\mcS(\alpha)$, and hence $\varphi_\mcS=f\circ \iota$. This proves part (1).

Now we prove part (2). For simplicity, we fix a constraint $\msC=(\msU,\msR)$, and without loss of generality assume $\msU=\{x_1,\ldots,x_t\}$. We first show that the commutation relation $[x_1,x_2]$ holds approximately in $\varphi_\mcS$. By construction, $Y_1$ and $Y_2$ commute, and $X_i$'s and $Y_i$'s are unitary. It follows that
 \begin{align*}
     X_1X_2\ket{\psi}&\approx_{\sqrt{4kM\epsilon}}X_1Y_2\ket{\psi}=Y_2X_1\ket{\psi}\\
     &\approx_{\sqrt{4kM\epsilon}} Y_2Y_1\ket{\psi}=Y_1Y_2\ket{\psi}\\
     &\approx_{\sqrt{4kM\epsilon}}Y_1X_2\ket{\psi}=X_2Y_1\ket{\psi}\\
     &\approx_{\sqrt{4kM\epsilon}} X_2X_1\ket{\psi}.
 \end{align*}
The above estimation implies $\varphi_S(r^*r)\leq 64kM\epsilon\leq T_\msB\cdot\epsilon$ for $r=[x_1,x_2]$. The calculation for the rest of the commutation relations follows similarly.

Now we consider the constraint relation $r_\msC:=\sum_{\phi:\phi(\msU)\not\in\msR}\prod_{i=1}^t\tfrac{1+\phi(x_j)x_i}{2}$. The probability that Alice responds with an unsatisfying assignment for the constraint $\mcC$ is $\leq kM\epsilon$, so
\begin{equation*}
\left\lVert\sum_{\phi:\phi(\msU)\notin\msR}\prod_{i=1}^t\frac{1+\phi(x_i)Y_i}{2}\ket{\psi}\right\rVert\leq \sqrt{kM\epsilon}.
\end{equation*}
Since there are at most $2^t\leq 2^M$ assignments for the variables in $\msU$, and all $\tfrac{1+\phi(x_i)X_i}{2}$'s and $\tfrac{1+\phi(x_i)Y_i}{2}$'s are projections,
\begin{align*}
\sum_{\phi:\phi(\msU)\notin\msR}\prod_{i=1}^t\frac{1+\phi(x_i)X_i}{2}\ket{\psi}&\approx_{2^M\sqrt{kM\epsilon}}  \sum_{\phi:\phi(\msU)\notin\msR}\left(\prod_{i=1}^{t-1}\frac{1+\phi(x_i)X_i}{2}\right)\frac{1+\phi(x_i)Y_t}{2}\ket{\psi}\\
&= \sum_{\phi:\phi(\msU)\notin\msR}\frac{1+\phi(x_i)Y_t}{2}\left(\prod_{i=1}^{t-1}\frac{1+\phi(x_i)X_i}{2}\right)\ket{\psi}\\
&\approx_{2^M\sqrt{kM\epsilon}} \cdots\\
&\approx_{2^M\sqrt{kM\epsilon}} \sum_{\phi:\phi(\msU)\notin\msR}\prod_{i=1}^t\frac{1+\phi(x_i)Y_i}{2}\ket{\psi}\approx_{\sqrt{kM\epsilon}}0.
\end{align*}
The above $\approx_{2^M\sqrt{kM\epsilon}}$ appears $t\leq M$ times. This implies
\begin{align*}
    \varphi(r_\msC^*r_\msC)\leq \left( (M2^M+1)\sqrt{kM\epsilon}\right)^2\leq \left( M2^{M+1}\sqrt{kM\epsilon}\right)^2=T_\msB\cdot\epsilon.
\end{align*}
We conclude that $\varphi_\mcS$ is a $(T_\msB\cdot\epsilon,\mcR(\msB))$-state.
\end{proof}

\section {Main result}\label{sec:Results}

\subsection{$\mcL$-families of BCS algebras and nonlocal games}\label{subsec:main-technical} 
Recall that for any boolean constraint system $\msB$, we can associate a BCS game $\mcG(\msB)$ and a game algebra $\mcA(\msB)$, as described in \Cref{sec:BCS}. The purpose of this section is to construct a reduction from the Halting problem for Turing machines to the (QC-Strict) problem for a certain family of BCS nonlocal games.

\begin{definition}\label{def:Lfamily}
Let $\mcL\subseteq\N$ be an RE set. A family of boolean constraint systems $\{\msB_m=(\mcX_m,\{\msC_i^{(m)}\})\}_{m\in\N}$ is called an \emph{$\mcL$-family} if there exists a variable $x_D\in\mcX_m$ for all $m\in \N$, and a sequence of positive integers $\{C_m\}_{m \in \N}$, satisfying:
\begin{enumerate}
    \item the mapping $m\mapsto (\msB_m,C_m)$ is computable;
    \item if $m\in\mcL$, then there exists a tracial state $\tau$ on $\mcA(\msB_m)$ such that
    \begin{equation*}
        \tau(D)>0;
    \end{equation*}
    \item if $m \notin \mcL$, then for every $\epsilon$-perfect strategy $\mcS$ for $\mcG(\msB_m)$, the associated state $\varphi_{\mcS}$ on $\C \Z_2^{*\mcX_m}$ satisfies
    \begin{equation*}
        \varphi_S(D)\leq C_m\cdot \epsilon.
    \end{equation*}
\end{enumerate}
Here $D:=\tfrac{1-x_D}{2}$.
\end{definition}

\begin{theorem}\label{thm:blackbox1}
    Let $\mcL\subseteq\N$ be an RE set and let $\{\msB_m\}_{m\in\N}$ be an $\mcL$-family of boolean systems. Let $x_D$ and $\{C_m\}_{m\in \N}$ be as in $\Cref{def:Lfamily}$. For every $m\in\N$, let $p_m=\tfrac{C_m}{1+C_m}$, and define the BCS game $\mcG_m$ consisting of two sub-games $\mcG_m^{(1)} $ and $\mcG_m^{(2)}$ as follows.
    \begin{itemize}
        \item With probability $p_m$, play $\mcG_m^{(1)}=\mcG(\msB_m)$.
        \item With probability $1-p_m$, play $\mcG_m^{(2)}$ in which the referee sends $x_D$ to both Alice and Bob, and they win if and only if they both respond with $-1$.
    \end{itemize}
Then $m\mapsto \mcG_m$ is a computable mapping such that
\begin{enumerate}
    \item if $m\in\mcL$, then $\omega_{qc}^{sync}(\mcG_m)>p_m$, and
    \item if $m\notin\mcL$, then $\omega_{qc}(\mcG_m)\leq p_m$.
\end{enumerate}
\end{theorem}

\begin{proof}
Since the mapping $m\mapsto (\msB_m,C_m)$ is computable, it is clear that the mapping  $m\mapsto \mcG_m$ is computable.

Now fix $m \in \N$, and write $\msB_m = \bigl(\mcX, \{\msC_i \}_{i=1}^k\bigr)$. In $\mcG_m$, Bob cannot tell which sub-game he is playing. So any $qc$-strategy 
    \begin{equation*}
        \mcS=\left(\ket{\psi}\in \mcH,\{P^{(i)}_{\phi}\}\cup\{A_D\},\{X_j\}  \right)
    \end{equation*}
    can be specified as follows. Alice uses PVMs $\{P^{(i)}_{\phi}\}$ for the constraints $\msC_i$, together with a binary observable $A_D$ for the variable $x_D$; 
Bob uses binary observables $\{X_1,\ldots,X_n\}$ for the boolean variables $\mcX=\{x_1,\ldots,x_n\}$.  

Suppose $m\in\mcL$. Let $\tau$ be a tracial state on $\mcA(\msB_m)$ as in part (1) of \Cref{def:Lfamily}. So $\tau(D)>0$, where $D=\tfrac{1-x_D}{2}$. Consider the GNS representation $\big(\mcH_\tau,\pi_\tau,\ket{\tau} \big)$ of $\tau$, and let $X_j:=\pi_{\tau}(x_j)$ for all $1\leq j\leq n$. In particular, we write $X_D$ for $\pi_\tau(x_D)$. Since $\tau$ is tracial, the left action $\pi_\tau$ admits a commuting right action $\pi_{\tau}^{op}$ of the opposite algebra $\mcA(\msB_m)^{op}$. Let $P^{(i)}_\phi=\prod_{x_j\in \msU_i}\tfrac{1}{2}(1+\phi(x_j)\pi_\tau^{op}(x_j^{op}))$ for all $\phi$, and let $A_D=\pi_\tau^{op}(x_D^{op})$. By  \Cref{prop:trace_perfect}, $\mcS:=\big(\ket{\tau}\in \mcH_\tau,\{P^{(i)}_{\phi}\}\cup\{A_D\},\{X_j\} \big)$ is a synchronous $qc$-strategy for $\mcG_m$ such that $\omega(\mcG_m^{(1)};\mcS)=1$. We also have
\begin{align*}
    \omega(\mcG_m^{(2)};\mcS)&=\bra{\tau}\tfrac{1-A_D}{2}\cdot\tfrac{1-X_D}{2}  \ket{\tau}\\
    &=\bra{\tau} \pi_\tau^{op}\left( \tfrac{1-x_D^{op}}{2}\right) \pi_\tau\left(\tfrac{1-x_D}{2} \right) \ket{\tau}\\
    &=\bra{\tau} \pi_\tau\left(D^2 \right) \ket{\tau}\\
    &=\tau(D)>0.
\end{align*}
It follows that
\begin{equation*}
    \omega(\mcG_m;\mcS)=p_m\cdot\omega(\mcG_m^{(1)};\mcS) + (1-p_m)\cdot\omega(\mcG_m^{(2)};\mcS)>p_m.
\end{equation*}
Hence $\omega_{qc}^{syn}(\mcG_m)>p_m$ when $m\in\mcL$.

Now suppose $m\notin \mcL$. For any $qc$-strategy 
    \begin{equation*}
         \mcS=\left(\ket{\psi}\in \mcH,\{P^{(i)}_{\phi}\}\cup\{A_D\},\{X_j\}  \right),
    \end{equation*}
let $\epsilon:=1-\omega(\mcG_m^{(1)};\mcS)$.
Let $\pi:\C\Z_2^{*\mcX}\arr \mcB(\mcH)$ be the representation sending $x_i\mapsto X_i$ for all $1\leq i\leq n$ and let $\varphi$ be the state on $\C\Z_2^{*\mcX}$ defined by $\varphi(\alpha)=\bra{\psi}\pi(\alpha)\ket{\psi}$. Then $\pi$ and $\varphi$ are the associated representation and state of $\mcS$ for the BCS game $\mcG_m^{(1)}$. By part (2) of \Cref{def:Lfamily}, $\varphi(P)\leq C_m\epsilon$.  Since $\frac{1-A_D}{2}$ is a projection commuting with the projection $\pi(D)$,
\begin{align*}
\omega(\mcG_m^{(2)};\mcS)=\bra{\psi}\frac{1-A_D}{2} \pi(D)\ket{\psi}\leq   \bra{\psi}\pi(D)\ket{\psi} =\varphi(D)\leq C_m\epsilon.
\end{align*}
It follows that
\begin{align*}
    \omega(\mcG_m;\mcS)&=p_m \cdot \omega(\mcG_m^{(1)};\mcS)+(1-p_m)\cdot \omega(\mcG_m^{(2)};\mcS)\\
    &\leq \frac{C_m}{1+C_m} (1-\epsilon) + \frac{1}{1+ C_m}\cdot C_m\cdot \epsilon\\
    & = \frac{C_m}{1+C_m} = p_m
\end{align*}
for any $qc$-strategy $\mcS$. Hence $\omega_{qc}(\mcG_m)\leq p_m$ whenever $m\notin \mcL$.
\end{proof}

\begin{remark}
From the proof, we see that one could make the sub-game $\mcG_m^{(2)}$ even more trivial by sending only the question $x_D$ to Bob, requiring the answer $-1$, and disregarding Alice's response. In that case, $\omega(\mcG_m^{(2)};\mcS)$ would simply be $\varphi_\mcS(D)$. However, the resulting game $\mcG_m$ will no longer be a BCS game. We chose the current BCS game formulation so that both players are involved. It better aligns with the algebraic framework and is more convenient for potential reductions to other problems. 
\end{remark}

\begin{theorem}\label{thm:Lfamily}
    For any RE set $\mcL\subseteq \N$, there exists an $\mcL$-family of boolean constraint systems.
\end{theorem}

We defer the proof of this theorem to \Cref{sec:Lfamily}.\\

The construction of an $\mcL$-family is based on the algebraic techniques for embedding Turing machines into $*$-algebras in \cite{MSZ23}, but additional work is required to adapt it to the present setting and to ensure that the reduction satisfies the stronger conditions in \Cref{def:Lfamily}. We remark that the $\mcL$-family of boolean constraint systems $\{\msB_m=(\mcX_m,\{\msC^{(m)}_i=(\msU_i^{(m)},\msR_i^{(m)})\})\}_{m\in\N}$ constructed in \Cref{sec:Lfamily} has the properties that $\abs{\mcX_m}=O(m)$ and $\abs{\msU_i^{(m)}}=O(1)$. So the resulting family of BCS games $\{\mcG_m\}_{m\in\N}$ has $O(\log m)$-bit questions and $O(1)$-bit answers.

\subsection{Undecidability results and the NPA hierarchy}

In this section we explain how our technical ingredients from \cref{subsec:main-technical} yield the two main theorems: $\RE$-hardness of determining whether a nonlocal game has commuting-operator value strictly greater than $1/2$ (\cref{thm:introI}), and the existence of a game $\mcG$ with $\omega^{(k)}_{\mathrm{npa}}(\mcG)>\omega_{qc}(\mcG)$ for all $k\in\N$ (\cref{thm:main_npa}). Both follow from the following more general \cref{thm:BCS-Sync-RE-hard}, which is an immediate consequence of  \cref{thm:blackbox1,thm:Lfamily}.

\begin{theorem}\label{thm:BCS-Sync-RE-hard}
     Given a BCS game $\mcG$ and a rational number $\theta\in (0,1)$ it is $\RE$-hard to decide whether:
    \begin{enumerate}
        \item $\omega^{sync}_{qc}(\mcG)> \theta$, or
        \item $\omega_{qc}(\mcG)\leq \theta$,
    \end{enumerate}
\end{theorem}

\begin{proof}
    Fix an $\RE$-hard set $\mcL \subseteq \mathbb{N}$. By \cref{thm:Lfamily} there exists a corresponding $\mcL$-family of boolean constraint systems $\lbrace \msB_m \rbrace$. Applying \cref{thm:blackbox1} gives a family of BCS games $\mcG_m$ and rational numbers $p_m$ satisfying  $\omega^{sync}_{qc} (\mcG_m) > p_m$ if and only if $m \in \mcL$.
\end{proof}

Since $\omega_{qc}(\mcG_m) \geq \omega^{sync}_{qc}(\mcG_m)$ \cref{thm:BCS-Sync-RE-hard} implies that the (QC-Strict-$\theta$) problem is $\RE$-hard. For convenience we restate this problem and the result in the following theorem.

\begin{theorem} The following problem is $\RE$-hard.\label{thm:strict_theta_RE-hard}
 \begin{equation*}
    \tag{QC-Strict-$\theta$}
    \begin{minipage}{.65\textwidth}
    Given a nonlocal game $\mcG$ and rational number $\theta\in (0,1)$ decide if $\omega_{qc}(\mcG)>\theta$.
    \end{minipage}
\end{equation*}
\end{theorem}

In particular, by combining with an always-win game or an always-lose game, the RE-hardness of (QC-Strict-$\theta$) implies the RE-hardness of (QC-Strict), i.e. the special case $\theta=\tfrac{1}{2}$.

\begin{proof}[Proof of \cref{thm:introI}]
    For $\theta\geq \frac{1}{2}$ (resp. $\theta< \frac{1}{2}$), define a new game $\mcG'$ as follows: with probability $\frac{1}{2\theta}$ (resp.~$\frac{1}{2(1-\theta)}$ ) the players play $\mcG$, and with probability $1-\frac{1}{2\theta}$ (resp. $1 - \frac{1}{2(1-\theta)}$) they play an always-lose game $\mcG_{lose}$ (resp. an always-win game $\mcG_{win}$). By construction, $\omega_{qc}(\mcG)>\theta$ if and only if $\omega_{qc}(\mcG')>\frac{1}{2}$. Since $\mcG\mapsto \mcG'$ is computable, we conclude that (QC-Strict) is RE-hard.
\end{proof}

 We do not know if the (QC-Strict-$\theta$) problem is contained in $\RE$.
Nevertheless, since the problem is $\RE$-hard there cannot be any $\coRE$-algorithm for the (QC-Strict-$\theta$) problem. It is precisely this tension that establishes our main result about the NPA hierarchy (\cref{thm:main_npa}).

\begin{proof}[Proof of \cref{thm:main_npa}]
    The proof is by contradiction. Suppose that for every nonlocal game $\mcG$ there exists a $k \in \mathbb{N}$ such that $\omega^{(k)}_{\mathrm{npa}}(\mcG) =\omega_{qc}(\mcG)$. Then, consider the following algorithm for the (QC-Strict-$\theta$) problem. Iterate over $k = 1,2,\ldots$, checking in sequence whether
\begin{equation}\label{eq:scov_coRE}
    \omega^{(k)}_{\mathrm{npa}}(\mathcal{G}) \leq \theta.
\end{equation}
There is an effective procedure to decide \cref{eq:scov_coRE} (see  \cref{subsec:prelim-NPA}) for each $k \in \mathbb{N}$. Hence, if $\omega_{qc}(\mcG) \leq \theta$, then from our assumption, there exists $k \in \mathbb{N}$ for which $\omega^{(k)}_{\mathrm{npa}}(\mathcal{G}) \leq \theta$, in which case the algorithm will halt and reject. In particular, such an algorithm shows that the (QC-Strict-$\theta$) problem is in $\coRE$. 
On the other hand, by \cref{thm:strict_theta_RE-hard} the (QC-Strict-$\theta$) problem is $\RE$-hard. Hence, for any language in $\RE$ there is a reduction to the (QC-Strict-$\theta$), and since (QC-Strict-$\theta$) is in $\coRE$ this implies that $\RE\subseteq \coRE$, a contradiction.
\end{proof}

The proof of \Cref{thm:main_npa} is not constructive. Nevertheless, one can use the mapping from \cref{thm:blackbox1} to give an explicit example of a game $\mcG$ for which $\omega^{(k)}_{\mathrm{npa}}(\mcG)>\omega_{qc}(\mcG)$ for all $k\in \N$. To do this, first recall that by picking an enumeration of Turing machines, we can take the $\RE$ set $\mcL$ to be in correspondence with Turing machines which halt and accept on the empty tape. In particular, \cref{thm:blackbox1} together with \cref{thm:Lfamily} gives a computable mapping $M \mapsto \mcG_M$, which given the description of a Turing machine $M$, outputs the description of the nonlocal game $\mcG_M$ and rational number $\theta_M$ satisfying $\omega_{qc}(\mcG_M) > \theta_M$ if and only if $M$ halts and accepts on the empty tape. From this, it is possible to define a Turing machine $M_0$ which does not halt and for which the corresponding BCS nonlocal game $\mcG_{M_0}$ satisfies
\[
   \omega_{qc}(\mcG_{M_0}) = \theta_{M_0}
   \quad\text{while}\quad 
   \omega^{(k)}_{\mathrm{npa}}(\mathcal{G}_{M_0}) > \theta_{M_0} \;\; \text{for all } k \in \mathbb{N}.
\]
In more detail, let $M_0$ be the Turing machine which on any input, computes a description of the corresponding game $\mathcal{G}_{M_0}$ and rational number $\theta_{M_0}$, and then iterates over $k = 1,2,\ldots$, checking for each $k$ whether
\[
    \omega^{(k)}_{\mathrm{npa}}(\mathcal{G}_{M_0}) \leq \theta_{M_0}.
\]
The machine $M_0$ halts at the first $k$ for which this inequality is satisfied. Suppose $M_0$ were to halt. Then by the properties of the computable map $M \mapsto \mcG_M$, it must be the case that $\omega_{qc}(\mcG_{M_0}) > \theta_{M_0}$ and hence $\omega^{(k)}_{\mathrm{npa}}(\mathcal{G}_{M_0}) > \theta_{M_0}$ for all $k \in \mathbb{N}$. However, this contradicts the halting condition of ${M_0}$ which requires $\omega_{npa}^{(k)}(\mcG_{M_0}) \leq \theta_{M_0}$ for some $k$. Therefore ${M_0}$ does not halt, and by the properties of the mapping $M \mapsto \mcG_M$ we see that $\omega_{qc}(\mcG_{M_0}) \leq \theta_{M_0}$. Furthermore, since ${M_0}$ never satisfies its halting condition, we must have $\omega^{(k)}_{\mathrm{npa}}(\mathcal{G}_{M_0}) > \theta_{M_0}$ for all $k \in \mathbb{N}$. Lastly, since $\omega^{(k)}_{\mathrm{npa}}(\mathcal{G}_{M_0})$ converges to $\omega_{qc}(\mcG_{M_0})$, it must be the case that $\omega_{qc}(\mcG_{M_0})=\theta_{M_0}$.

\section{Construction of an \texorpdfstring{$\mcL$}{L}-family}\label{sec:Lfamily}

\subsection{A quantitative embedding theorem for nested conjugacy BCS relations}\label{subsec:NestedConjugacy}

\begin{definition}
    Given $\mcX=\{x_1,\ldots,x_n\}$, a \emph{nested conjugacy monomial} over $\mcX$ of depth $\ell\geq 0$ is a monomial of the form
    \begin{equation*}
        x_{i_\ell}x_{i_{\ell-1}}\cdots x_{i_1}x_{i_0}x_{i_1}\cdots x_{i_{\ell-1}}x_{i_\ell}, 
    \end{equation*}
    where each $x_{i_j}\in\mcX$, and $x_{i_j}\neq x_{i_{j-1}}$ for all $1\leq j\leq \ell$ (for $\ell=0$, the monomial is simply $x_{i_0}$). 
     We denote by $\mcN^{\text{mon}}_\ell(\mcX)$ the set of nested conjugacy monomials over $\mcX$ of depth $\leq \ell$, and define
    \begin{equation*}
        \mcN^{\text{mon}}(\mcX):=\bigcup_{\ell\geq 0}\mcN^{\text{mon}}_\ell(\mcX)
    \end{equation*}
    to be the set of all nested conjugacy monomials over $\mcX$.
\end{definition}

Suppose $\mcX=\{x_1,\ldots,x_n\}$ is a set of boolean variables, and let  $w$ be a nested conjugacy monomial over $\mcX$ of depth $\ell\geq 1$. Recall that we also regard $x_1,\ldots,x_n$ as canonical generators of the $*$-algebras $\C^*\ang{\mcX}$ and $\C\Z_2^{*\mcX}$, so $w$ is a monomial in $\C^*\ang{\mcX}$ of degree $2\ell+1$, and it is straightforward to verify that $w$ is a unitary of order 2 in $\C\Z_2^{*\mcX}$. As such, we may formally treat each such monomial $w$ as a new boolean variable.

\begin{definition}\label{def:nestedvariable}
    Given a set of boolean variables $\mcX=\{x_1,\ldots,x_n\}$, we define the set of nested conjugacy variables over $\mcX$ by
    \begin{equation*}
        \mcN^{\text{var}}(\mcX):=\{w(i_0,i_1,\ldots,i_\ell):1\leq i_j\leq n, i_{j}\neq i_{j+1}  \text{ for all }j  \},
    \end{equation*}
    where $w(i_0)=x_{i_0}$, and each $w(i_0,i_1,\dots,i_\ell)$ for $\ell\geq 1$ is a new boolean variable. We denote by $\Psi_\mcX$ the bijection from $\mcN^{\text{var}}(\mcX)\arr\mcN^{\text{mon}}(\mcX)$
    sending 
    \begin{equation*}
        w(i_0,i_1,\ldots,i_\ell)\mapsto x_{i_\ell}\cdots x_{i_1}x_{i_0}x_{i_1}\cdots x_{i_\ell}.
    \end{equation*} 
    For any nested conjugacy variable $w\in \mcN^{\text{var}}(\mcX)$, the depth of $w$ is the depth of the nested conjugacy monomial $\Psi_\mcX(w)$. For every $\ell\geq 0$, we denote by $\mcN_\ell^{\text{var}}(\mcX)$ the set of nested conjugacy variables over $\mcX$ of depth $\leq \ell$. 
\end{definition}

Let $\mcX$ be a set of boolean variables and let $\mcW\subseteq \mcN^{\text{var}}(\mcX)$ be a finite subset of nested conjugacy variables over $\mcX$. It is natural to consider boolean constraint systems over boolean variables $\mcW$. Given such a system $\msB=(\mcW,\{\msC_i\})$, its associated BCS relations 
$\mcR(\msB)$ (as defined in \Cref{def:bcs_alg}) are $*$-polynomials in $\C^*\ang{\mcW}$. Again, we regard every $w\in\mcW$ as a canonical generator of $\C^*\ang{\mcW}$, so the mapping $\Psi_\mcX$ defined above sends $w$ to the corresponding nested conjugacy monomial in $\C^*\ang{\mcX}$. It follows that $\Psi_\mcX$ naturally extends to a $*$-homomorphism from $\C^*\ang{\mcW}\arr\C^*\ang{\mcX}$, by linearity and multiplicativity. In the rest of the paper, we keep using $\Psi_\mcX$ to denote this $*$-homomorphism. Applying $\Psi_{\mcX}$ to the BCS relations in $\mcR(\msB)$ yields a set of $*$-polynomials in $\C^*\ang{\mcX}$, which we call the nested conjugacy BCS relations.

\begin{definition}
    Let $\mcX$ be a set of boolean variables and let $\mcW\subseteq \mcN^{\text{var}}(\mcX)$ be a finite subset. Let $\msC=(\msU,\msR)$ be a boolean constraint with $\msU\subseteq\mcW$. The \emph{nested conjugacy BCS relations} associated with $\msC$ is defined as
    \begin{equation*}
        \mcR^{\mcX}_{\text{nest}}(\msC): =\{\Psi_\mcX(r):r\in \mcR(\msC) \}\subseteq \C^*\ang{\mcX}.
    \end{equation*}
    For any boolean constraint system $\msB=(\mcW,\{\msC_i\}_{i=1}^k)$, its associated \emph{nested conjugacy BCS relations} is
    \begin{equation*}
        \mcR^{\mcX}_{\text{nest}}(\msB):= \bigcup_{i=1}^k \mcR^{\mcX}_{\text{nest}}(\msC_i).
    \end{equation*}
    The \emph{nested conjugacy BCS algebra} associated with $\msB$ is the quotient 
    \begin{equation*}
        \mcA^{\mcX}_{\text{nest}}(\msB):=\C\Z_2^{*\mcX}/\ang{\mcR^{\mcX}_{\text{nest}}(\msB)}.
    \end{equation*}
\end{definition}

The superscript $\mcX$ in $\mcA^{\mcX}_{\text{nest}}(\msB)$ emphasizes that this algebra is a quotient of $\C\Z_2^{*\mcX}$ rather than $\C\Z_2^{*\mcW}$.

\begin{example}\label{example:nestedBCS}
    Suppose we have a commutation relation $[x_1x_2,x_3]=0$ between $x_3$ and the product of two variables $x_1,x_2$. In some cases we do not want $x_1$ and $x_2$ to commute, so we do not want to encode this relation using a boolean constraint that involves both $x_1$ and $x_2$. Instead, we introduce the nested conjugacy variable $\omega(3,2,1):=\Psi_\mcX^{-1}(x_1x_2x_3x_2x_1)$ and consider the constraint $\msC:\omega(3,2,1)=x_3$ over nested conjugacy variables. The associated nested conjugacy BCS relations are $\{x_1x_2x_3x_2x_1=x_3,[x_1x_2x_3x_2x_1,x_3]=0\}$, which is equivalent to $[x_1x_2,x_3]=0$.
\end{example}

\begin{theorem}\label{thm:blackbox2}
Let $\mcX$ be a set of boolean variables. Let $\msB=(\mcW,\{\msC_i=(\msU_i,\msR_i)\})$ be a boolean constraint system, where $\mcW\subseteq\mcN_\ell^{\text{var}}(\mcX)$ for some $\ell\geq 1$, and let $M:=\max_i\abs{\msU_i}$. Let $ \mcA^{\mcX}_{\text{nest}}(\msB)=\C\Z_2^{*\mcX}/\ang{\mcR^{\mcX}_{\text{nest}}(\msB)}$ be the associated nested conjugacy BCS algebra. Then there exists a boolean constraint system $\wtd{\msB}$ over $\wtd{\mcX}$ with $\mcX\subseteq \wtd{\mcX}$ such that the following holds.
    \begin{enumerate}
        \item The natural $*$-homomorphism $\C^*\ang{\mcX} \to \C^*\ang{\wtd{\mcX}}$ sending $x \mapsto x$ for all  $x \in \mcX$  descends to an embedding from $\mcA^{\mcX}_{\text{nest}}(\msB)\hookrightarrow \mcA(\wtd{\msB})$, the BCS algebra associated with $\wtd{\msB}$.
      \item Any tracial state on $\mcA^{\mcX}_{\text{nest}}(\msB)$ extends to a tracial state on $\mcA(\wtd{\msB})$.
      \item For any $\beta\in \C^*\ang{\mcX}$, if $\beta$ is trivial in $\mcA^{\mcX}_{\text{nest}}(\msB)$ and has an $\mcR^{\mcX}_{\text{nest}}(\msB)$-decomposition in $\C\Z_2^{\mcX}$ of size $\Lambda$, then as an element of $\C^*\ang{\wtd{\mcX}}$, $\beta$ has an $\mcR(\wtd{\msB})$-decomposition in $\C\Z_2^{*\wtd{\mcX}}$ of size $\leq 2^{16}M^2\ell^2\Lambda$.
    \end{enumerate}

\end{theorem}

Part (3) states that the blow-up of the $\mcR$-decomposition under this embedding only depends (quadratically) on $M$ and $\ell$, where $M$ is the maximal size of contexts in $\msB$, and $\ell$ is an upper bound of the depth of nested conjugacy variables in $\mcW$. 

We prove this theorem in two steps of embedding. First, we replace each nested conjugacy BCS relation with one BCS relation and several conjugacy relations (\Cref{lemma:flatBCS}). Then we further embed every conjugacy relation into BCS relations (\Cref{lemma:LCSembedding}).

\begin{definition}\label{def:flatBCS}
    Let $\mcX$ be a set of boolean variables. For every nested conjugacy variable $w:=w(i_0,i_1,\ldots,i_\ell)\in \mcN^{\text{var}}(\mcX)$ of depth $\ell\geq 1$, we define the associated set of flattened variables
    \begin{equation*}
        \mcV(w):=\{    w(i_0,i_1,\ldots,i_j): 1\leq j\leq \ell   \}\subseteq\mcN_\ell^{\text{var}}(\mcX),
    \end{equation*}
    and the associated conjugacy relations $\mcR_{\text{conj}}(w)\subseteq\C^*\ang{\mcX\sqcup\mcV(w)}$, consisting of relations
\begin{equation*}
   x_{i_j}w(i_0,i_1,\ldots,i_{j-1})x_{i_j}- w(i_0,i_1,\ldots,i_j), \text{ for } 1\leq j\leq \ell.
\end{equation*}
For $w\in \mcN^{\text{var}}_0(\mcX)=\mcX$, we let $\mcV(w)$ and $\mcR_{\text{conj}}(w)$ be empty.

Let $\mcW\subseteq \mcN^{\text{var}}(\mcX)$ be a finite subset. For any boolean constraint system $\msB=(\mcW,\{\msC_i\})$, we define the associated flat conjugacy BCS algebra as the quotient
\begin{equation*}
    \mcA_{\text{flat}}(\msB):=\C\Z_2^{*\mcW_{\text{flat}}}/\ang{\mcR_{\text{flat}}(\msB)},
\end{equation*}
where
    \begin{equation*}
        \mcW_{\text{flat}}:=\mcX\cup\left(\bigcup_{w\in\mcW} \mcV(w) \right),
    \end{equation*}
    and
  \begin{equation*}
      \mcR_{\text{flat}}(\msB):=\mcR(\msB)\cup\left(\bigcup_{w\in\mcW} \mcR_{\text{conj}}(w) \right).
  \end{equation*}  
  Here $\mcR(\msB)\subseteq \C^*\ang{\mcW}$ is the set of BCS relations associated with $\msB$.
\end{definition}

Note that in the above definition, every $w\in\mcW$ itself is in $\mcV(w)$, so both $\mcX$ and $\mcW$ are contained in $\mcW_{\text{flat}}$, and hence $\mcR(\msB)$ and $ \mcR_{\text{conj}}(w)$ are subsets of $\C^*\ang{\mcW_{\text{flat}}}$. The following proposition illustrates that a nested conjugacy BCS algebra  $\mcA^{\mcX}_{\text{nest}}(\msB)$ naturally embeds into the flat conjugacy BCS algebra $\mcA_{\text{flat}}(\msB)$.

\begin{lemma}\label{lemma:flatBCS}
    Let $\mcX$ be a set of boolean variables. Let $\msB=(\mcW,\{\msC_i=(\msU_i,\msR_i)\})$ be a boolean constraint system, where $\mcW\subseteq\mcN_\ell^{\text{var}}(\mcX)$ for some $\ell\geq 1$, and let $M:=\max_i\abs{\msU_i}$. Let  
    $\mcA^{\mcX}_{\text{nest}}(\msB)=\C\Z_2^{*\mcX}/\ang{\mcR^{\mcX}_{\text{nest}}(\msB)}$ and $\mcA_{\text{flat}}(\msB)=\C\Z_2^{*\mcW_{\text{flat}}}/\ang{\mcR_{\text{flat}}(\msB)}$ be the nested conjugacy BCS algebra and flat conjugacy BCS algebra associated with $\msB$, respectively. Then the following holds.
     \begin{enumerate}
        \item The natural $*$-homomorphism $\C^*\ang{\mcX} \to \C^*\ang{\mcW_{\text{flat}}} $ sending $x \mapsto x$ for all  $x \in \mcX$  descends to an embedding from $\mcA^{\mcX}_{\text{nest}}(\msB)\hookrightarrow \mcA_{\text{flat}}(\msB)$.
      \item Any tracial state on $\mcA^{\mcX}_{\text{nest}}(\msB)$ extends to a tracial state on $\mcA_{\text{flat}}(\msB)$.
      \item For any $\beta\in \C^*\ang{\mcX}$, if $\beta$ has an $\mcR^{\mcX}_{\text{nest}}(\msB)$-decomposition in $\C\Z_2^{\mcX}$ of size $\Lambda$, then $\beta$ has an $\mcR_{\text{flat}}(\msB)$-decomposition in $\C\Z_2^{*\mcW_{\text{flat}}}$ of size $\leq 9M^2\ell^2 \Lambda$.
    \end{enumerate}
\end{lemma}

\begin{proof}
    For any $w:=w(i_0,i_1,\ldots,i_k)\in \mcW$ of depth $k\geq 2$,
    \begin{align*}
        &\Psi_\mcX(w)-w\\
        & = x_{i_k}\cdots x_{i_1}x_{i_0}x_{i_1}\cdots x_{i_k}-w(i_0;i_1,\ldots,i_k)\\
        & = x_{i_k}w(i_0;i_1,\ldots,i_{k-1})x_{i_k}-w(i_0,i_1,\ldots,i_k)\\
        &\quad + \sum_{t=2}^kx_{i_k}\cdots x_{i_t}\left( x_{i_{t-1}} w(i_0,i_1,\ldots,i_{t-2}) x_{i_{t-1}}-w(i_0,i_1,\ldots,i_{t-1})\right)x_{i_t}\cdots x_{i_k}. 
    \end{align*}
Since 
\begin{equation*}
    x_{i_{t-1}} w(i_0,i_1,\ldots,i_{t-2}) x_{i_{t-1}}-w(i_0,i_1,\ldots,i_{t-1}) \text{ for } 2\leq t\leq k+1
\end{equation*}
are all relations in $\mcR_{\text{conj}}(w)\subseteq \mcR_{\text{flat}}(\msB)$, and they each have operator norm $2$ in $\C\Z_2^{*\mcW_{\text{flat}}}$, it follows that the above $\mcR_{\text{flat}}(\msB)$-decomposition of $\Psi_\mcX(w)-w$ in $\C\Z_2^{*\mcW_{\text{flat}}}$ has size 
\begin{equation*}
   \sum_{t=0}^{k-1}(1+2t)=k^2.
\end{equation*}
The same arguments hold for $w\in \mcW$ of depth $k=1$ and $0$. Since every nested conjugacy variable in $\mcW$ has depth $\leq \ell$, we conclude that $\Phi_\mcX(w)-w$ has an $\mcR_{\text{flat}}(\msB)$-decomposition in $\C\Z_2^{*\mcW_{\text{flat}}}$ of size $\leq \ell^2$ for every $w\in \mcW$. By iteratively applying part (2) of \Cref{lemma:sizecalc}, we have that $\Psi_\mcX(\alpha)-\alpha$ has an $\mcR_{\text{flat}}(\msB)$-decomposition in $\C\Z_2^{*\mcW_{\text{flat}}}$ of size $\leq \deg(\alpha)^2\ell^2$ for every $*$-monomial $\alpha$ in $\C^*\ang{\mcW}$.

For any $\Psi_\mcX(r)\in \mcR^{\mcX}_{\text{nest}}(\msB)$, $r$ is a $*$-polynomial in $\C^*\ang{\mcW}$ of the form $r=\sum_i c_ip_i$ where every $p_i$ is a $*$-monomial and $c_i\in\C$. If $r$ is a commutator, then $\deg(p_i)\leq 2$ and $\sum_i\abs{c_i}=2$. If $r$ is a constraint relation, then $\deg(p_i)\leq M$ and $\sum_i\abs{c_i}\leq 1$. It follows from part (1) of \Cref{lemma:sizecalc} that
\begin{equation*}
    \Psi_\mcX(r)-r=\sum_ic_i\left( \Psi_\mcX(p_i)-p_i \right)
\end{equation*}
    has an $\mcR_{\text{flat}}(\msB)$-decomposition in $\C\Z_2^{*\mcW_{\text{flat}}}$ of size 
    \begin{equation*}
        \sum \abs{c_i}\deg(p_i)^2\ell^2\leq \left(\sum_i \abs{c_i}\right)M^2\ell^2\leq 2M^2\ell^2.
    \end{equation*}
    Since $r\in\mcR(\msB)\subseteq \mcR_{\text{flat}}(\msB)$, we conclude that $\Psi_\mcX(r)=r+\left(\Psi_\mcX(r)-r\right)$ has an $\mcR_{\text{flat}}(\msB)$-decomposition in $\C\Z_2^{*\mcW_{\text{flat}}}$ of size $\leq 2M^2\ell^2+1\leq 3M^2\ell^2$ for every $\Psi_\mcX(r)\in \mcR_{\text{nest}}(\msB)$. Part (1) follows immediately. 
    
    For part (3), note that $\norm{\Psi_\mcX(r)}_{\C\Z_2^{*\mcW_{\text{flat}}}}\geq 1$ for all $r\in \mcR^\mcX_{\text{nest}}(\msB)$ and that $\norm{r}_{\C\Z_2^{*\mcX_{\text{flat}}}}\leq 2$ for all $r\in \mcR_{\text{flat}}(\msB)$. Suppose $\beta\in \C^*\ang{\mcX}$ is trivial in $\mcA^{\mcX}_{\text{nest}}(\msB)$ and has an $\mcR^{\mcX}_{\text{nest}}(\msB)$-decomposition in $\C\Z_2^{\mcX}$ of size $\Lambda$, then by \Cref{lemma:Rdecomp}, $\beta$ has an $\mcR_{\text{flat}}(\msB)$-decomposition in $\C\Z_2^{*\mcW_{\text{flat}}}$ of size $\leq (1+2)3M^2\ell^2\Lambda=9M^2\ell^2\Lambda$.
    
    Now we prove part (2). Suppose $\tau$ is a tracial state on $\mcA_{\text{nest}}^\mcX(\msB)$, and let $(\mcH,\pi,\ket{\psi})$ be a GNS representation of $\tau$. Let $\hat{\pi}:\C^*\ang{\mcW_{\text{flat}}}\arr \mcB(\mcH)$ be the $*$-representation sending 
    \begin{equation*}
        w\mapsto \pi\left(\Psi_\mcX(w) \right)
    \end{equation*}
    for all $w\in \mcW_{\text{flat}}$. So in particular, $\hat{\pi}(x)=\pi(x)$ for all $x\in \mcX$. It is straightforward to verify that $\hat{\pi}(w)$ is a unitary of order $2$ for every $w\in \mcW_{\text{flat}}$ and that $\hat{\pi}(r)=0$ for all relations in $\mcR_{\text{flat}}(\msB)$. Hence $\hat{\pi}$ induces a representation $\wtd{\pi}:\mcA_{\text{flat}}(\msB)\arr \mcB(\mcH)$ such that $\wtd{\pi}\circ q =\hat{\pi}$ where $q$ is the canonical quotient map from $\C^*\ang{\mcW_{\text{flat}}}\arr \mcA_{\text{flat}}(\msB)$. Moreover, by construction, $\wtd{\pi}(w)$ is in the von Neumann algebra $\pi\left( \mcA_{\text{nest}}^\mcX(\msB)\right)''$ for all $w\in \mcW_{\text{flat}}$. This implies that the state $\wtd{\tau}$ on $\mcA_{\text{flat}}(\msB)$ defined by $\wtd{\tau}(\alpha)=\bra{\psi}\wtd{\pi}(\alpha)\ket{\psi}$ is a tracial state. Since for every $\beta\in \mcA_{\text{next}}^{\mcX}$,
    \begin{align*}
      \wtd{\tau}(\beta)=\bra{\psi}\wtd{\pi}(\beta)\ket{\psi}=\bra{\psi}\hat{\pi}\left(q(\beta)\right)\ket{\psi}=\bra{\psi}\pi(\beta)\ket{\psi}=\tau(\beta), 
    \end{align*}
    we conclude that $\wtd{\tau}$ is an extension of $\tau$.
\end{proof}

\begin{lemma}\label{lemma:embedding}
    Let $r=x_1x_2x_1-x_3$ be a conjugacy relation over variables $\{x_1,x_2,x_3\}$, and let $\mcP_r:= \C\Z_2^{*\{x_1,x_2,x_3\}}/\ang{r}$. There is a set of boolean variables $\mcX_r$ and a boolean constraint system $\msB_r$ over $ \{x_1,x_2,x_3\}\sqcup \mcX_r$ such that the following holds.
    \begin{enumerate}
        \item The $*$-homomorphism from $\C^*\ang{x_1,x_2,x_3}\arr\C^*\ang{\{x_1,x_2,x_3\}\sqcup \mcX_r}$ sending $x_i\mapsto x_i$, $i=1,2,3$, descends to an embedding from $\mcP_r\hookrightarrow\mcA(\msB_r)$, where 
        $\mcA(\msB_r)=\C\Z_2^{\{x_1,x_2,x_3\}\sqcup \mcX_r}/\ang{\mcR(\msB_r)}$ is the BCS algebra associated with $\msB_r$.
        \item Any tracial state on $\mcP_r$ extends to a tracial state on $\mcA(\msB_r)$.
        \item $r$ has an $\mcR(\msB_r)$-decomposition in $\C\Z_2^{*\{x_1,x_2,x_3\}\sqcup \mcX_r}$ of size at most $1750$.
    \end{enumerate}
\end{lemma}

The proof is based on the embedding for \emph{linear plus conjugacy groups} in \cite{S19} extended to the $*$-algebra setting.

\begin{proof}
We begin with the proof of (1), which proceeds in two steps. First, we will embed $\mcP_r$ into an algebra $\mcA_0$, and then we embed $\mcA_0$ into $\mcA$. We then argue that there exists a BCS $\msB_r$ for which $\mcA(\msB_r)=\mcA$. Let $\mcX=\{x_1,x_2,x_3\}$, and $\mcY_0=\{y_1,y_2,y_3\}\cup \{w_1,w_2,w_3\}\cup \{f,g\}$. Let $\mcA_0$ to be the $*$-algebra generated by $\{\mcX\cup \mcY_0\}$ over $\Z_2$ with relations:
\begin{enumerate}[(i)]
    \item $x_iy_iz_i-1$, $x_iy_i-y_ix_i$, $x_iz_i-z_ix_i$, and $y_iz_i-z_iy_i$ for all $1\leq i \leq 3$,
    \item $x_ifw_i-1$, $x_if-fx_i$, $x_iw_i-w_ix_i$, and $w_if-fw_i$, for $1\leq i \leq 3$,
    \item $gy_2z_3-1$, $y_2z_3-z_3y_2$,
    \item $fy_1f-z_1$, $fy_2f-z_2$, $fy_3f-z_3$, and $w_1y_2w_1-z_3$.
\end{enumerate}
Using the relations (i)-(iv) we see that \begin{align*}
    x_1x_2x_1-x_3&=(fw_1)(y_2z_2)(fw_1)-(y_3z_3)\\&=(fw_1y_2w_1f)(fw_1z_2w_1f)-y_3z_3\\
    &=(fz_3f)(fy_3f)-y_3z_3\\
    &=y_3z_3-y_3z_3\\
    &=0,
\end{align*}
hence the $*$-homomorphism $\phi_0:\C^*\langle\mcX\rangle\to \C^*\langle \mcX\cup\mcY_0\rangle$ sending $x_i\mapsto x_i$ for $1\leq i \leq 3$, descends to an embedding $\mcP_r\hookrightarrow \mcA_0$. At this point, we note that the relations (C0)$=\{\text{(i)-(iii)}\}$ in $\mcA_0$ are \emph{linear} BCS relations, see \cref{ex:lcs_alg}. However, $\mcA_0$ is not necessarily a BCS algebra due to the remaining conjugacy relations (iv). The next step in our proof will be to embed these relations into BCS relations.

We now construct the second embedding. Consider the relations in (iv) above. It consists of the  conjugacy relations $r_1:=fy_1f-z_1$, $r_2:=fy_2f-z_2$, $r_3:=fy_3f-z_3$, and $r_4:=w_1y_2w_1-z_3$. For convenience, we let $a_j,b_j,$ and $c_j$ denote the variables from $r_j$, so that $a_jb_ja_j-c_j$ for $1\leq j\leq 4$. For example, in $r_1$ we have $a_1=f$, $b_1=y_1$, and $c_1=z_1$. Let $\mcD_0=\{d_{j\ell}:1\leq j\leq 4,1\leq \ell \leq 7\}$, and then define $\mcA$ to be the finitely presented $*$-algebra over $\Z_2$ with generators $\mcX\cup \mcY_0 \cup \mcD_0$ and relations:
\begin{enumerate}
    \item[(A0)] (C0) (the relations (i)-(iii) from $\mcA_0$),
    \item[(A1)] $a_jd_{j1}d_{j2}-1,b_jd_{j2}d_{j3}-1,d_{j3}d_{j4}d_{j5}-1,a_jd_{j5}d_{j6}-1,c_jd_{j6}d_{j7}-1$, and $d_{j1}d_{j4}d_{j7}-1$ for all $1\leq j \leq 4$, and
    \item[(A2)]  $st-ts$, for each pair of distinct generators  $s$ and $t$ contained in each of the $6\times 4=24$ relations in (A1).
\end{enumerate}
The embedding $\phi_2:\C^*\langle \mcX\cup \mcY_0\rangle \to \C^*\langle\mcX_0\cup \mcY_0\cup \mcD_0 \rangle$ sending $y\mapsto y$ for all $y\in \mcX\cup \mcY_0$, descends to a $*$-homomorphism $\mcA_0\hookrightarrow\mcA$. To see this, it suffices to show that each relation $r_j$ for $1\leq j\leq 4$ from (iv) in $\mcA_0$ also holds in $\mcA$, as the remaining relations (C0) of $\mcA_0$ are the (A0) relations of $\mcA$. With this in mind, we see that
\begin{align*}
    a_jb_ja_j-c_j&=(d_{j1}d_{j2})(d_{j2}d_{j3})(d_{j5}d_{j6})-(d_{j6}d_{j7})\\
    &=d_{j1}(d_{j3}d_{j5})d_{j6}-(d_{j7}d_{j6})\\
    &=(d_{j1}d_{j4})d_{j6}-d_{j7}d_{j6}\\
    &=d_{j7}d_{j6}-d_{j7}d_{j6}\\
    &=0
\end{align*}
for $1\leq j \leq 4$, establishing the claim. The result now follows by composing the inclusions $\phi_1\circ\phi_0$, from which we obtain a $*$-homomorphism $\mcP_r\hookrightarrow \mcA$. To complete the proof of (1), we remark that the relations (A1)$\cup$(A2), and (A0)$=$(C0) are all linear relations. Hence, there is a (linear) boolean constraints system $\msB_r$ such that $\mcA(\msB_r)=\mcA$, where $\mcX_r=\{\mcY_0\cup \mcD_0\}$.

For (2), we observe that if $\pi$ is a representation of $\mcP_r$ on $\mcH$ then $\pi$ determines a representation $\varphi$ of $\mcA_0$ on $\mcH\otimes \C^2$ via:
\begin{align*}
    &\varphi(x_i)=\begin{pmatrix}
        \pi(x_i) & 0 \\ 0 & \pi(x_i)
     \end{pmatrix},
    \varphi(y_i)=\begin{pmatrix}
        \pi(x_i) & 0 \\ 0 & \Id
    \end{pmatrix}, \varphi(z_i)=\begin{pmatrix}
        \Id & 0 \\ 0 & \pi(x_i)
    \end{pmatrix},\\
    &\varphi(w_i)=\begin{pmatrix}
        0 & \pi(x_i) \\ \pi(x_i) & 0
    \end{pmatrix}
    \text{for }1\leq i \leq 3,\text{ and } \varphi(f)=\begin{pmatrix}
        0 & \Id \\ \Id & 0
    \end{pmatrix}, \varphi(g)=\begin{pmatrix}
        \Id & 0 \\ 0 & \Id
    \end{pmatrix}.
\end{align*}
Furthermore, the representation $\varphi:\mcP_r\to \mcA_0$ can be extended to are representation of $\mcA$ on $\mcH\otimes \C^4$ via:
\begin{align*}
        &\vartheta(x)=\begin{pmatrix}
        \varphi(x) & 0 \\ 0 & \varphi(x)
    \end{pmatrix} \text{ for all $x\in \mcX\cup \mcY_0$},\\
    &\vartheta(d_{j1})=\begin{pmatrix}
        0 & \varphi(a_j) \\ \varphi(a_j) & 0
    \end{pmatrix}, \vartheta(d_{j2})=\begin{pmatrix}
        0 & \Id \\ \Id & 0
    \end{pmatrix},\vartheta(d_{j3})=\begin{pmatrix}
        0 & \varphi(b_j) \\ \varphi(b_j) & 0
    \end{pmatrix},\\
    &\vartheta(d_{j4})=\begin{pmatrix}
        0 & \varphi(a_jb_j) \\ \varphi(a_jb_j) & 0
    \end{pmatrix}, \vartheta(d_{j5})=\begin{pmatrix}
        \varphi(b_ja_jb_j) & 0 \\ 0 & \varphi(a_j)
    \end{pmatrix},\text{ and }\\
    &\vartheta(d_{j6})=\begin{pmatrix}
        \varphi(b_jc_j) & 0 \\ 0 & \Id
    \end{pmatrix},
    \vartheta(d_{j7})=\begin{pmatrix}
        \varphi(b_j) & 0 \\ 0 & \varphi(c_j)
    \end{pmatrix}\text{ for } 1\leq j \leq 4.
\end{align*}
In other words, $\mcA$ embeds into the $4\times 4$ matrix algebra with entries in $\mcP_r$. Recalling that $M_4(\mcP_r)\iso \mcP_r\otimes M_4(\C)$, we see that if $\tau$ is a trace on $\mcP_r$, and $\text{tr}$ is the (unique) normalized trace on $M_4(\C)$, then $\tau'=\tau \otimes \text{tr}$ extends to a tracial state on $\mcA(\msB_r)$.

To establish (3), we start with two key observations. Firstly, our proof that $x_1x_2x_1-x_3$ holds in $\mcA_0$ provides an $\mcR$-decomposition of $r$ in $\C\Z_2^{* \mcX\cup\mcY_0}$ of size at most $5(1+2\times2)=25$. Secondly, our proof that $a_jb_ja_j-c_j$ is trivial in $\mcA$ for $1\leq j\leq 4$, shows that each $r_j$ from (iv) has an $\mcR$-decomposition in $\C\Z_2^{* \mcX\cup\mcY_0\cup \mcY_1}$ of size at most $7(1+2\times2)=35$. Determining the size of these $\mcR$-decompositions is sufficient, as the other relations have trivial $\mcR$-decompositions. We note that $\|r\|_{\mcA_0}=\|r\|_\mcA=2$, hence by \cref{lemma:Rdecomp} we conclude that $r$ has an $\mcR$-decomposition in $\C\Z_2^{*\{x_1,x_2,x_3\}\sqcup \mcX_r}$ of size at most $2\times 35\times 25=1750$.
\end{proof}

\begin{remark}
The resulting BCS $\msB_r$ in \cref{lemma:embedding} with associated BCS algebra $\mcP_r\hookrightarrow\mcA(\msB_r)$ has $42$ boolean variables, and $31$ linear constraints, each one having a context of size no greater than $3$.
\end{remark}

Putting everything together, we are now ready to prove \Cref{thm:blackbox2}.
\begin{proof}[Proof of \Cref{thm:blackbox2}]
Let $\mcA_{flat}(\msB)=\C\Z_2^{*\mcW_{flat}}/\ang{\mcR_{flat}(\msB)}$ be the flat BCS algebra associated with $\msB=(\mcX,\{\msC_i\})$. Here $\mcW_{\text{flat}}=\mcX\cup\mcV_0$ and $\mcR_{\text{flat}}(\msB)=\mcR(\msB)\cup\mcR_0$, where $\mcV_0:=\left(\bigcup_{w\in\mcW} \mcV(w) \right)$ and $\mcR_0:=\left(\bigcup_{w\in\mcW} \mcR_{\text{conj}}(w) \right)$,
as described in \Cref{def:flatBCS}. 

For every conjugacy relation $r\in\mcR_0$, we write $r=x_1^{(r)}x_2^{(r)}x_1^{(r)}-x_3^{(r)}$, and let $\mcX_r$ and $\msB_r=\left(\{x_1^{(r)},x_2^{(r)},x_3^{(r)} \}\sqcup\mcX_r,\{\msC_i^{(r)}\}\right)$ be the set of boolean variables and boolean constraint system as in \Cref{lemma:embedding}. Let 
\begin{equation*}
    \wtd{\mcX}:=\mcW_{\text{flat}}\sqcup \left(\bigsqcup_{r\in\mcR_0} \mcX_r 
    \right),
\end{equation*}
and define the boolean constraint system
\begin{equation*}
    \wtd{\msB}:=\left(\wtd{\mcX},\{\msC_i\}\cup\left(\bigcup_{r\in\mcR_0}\{\msC_i^{(r)}\}  \right) \right).
\end{equation*}
So $\mcR(\wtd{\msB})=\mcR(\msB)\cup\left(\bigcup_{r\in\mcR_0}\mcR(\msB_r)\right)$. By iteratively applying \Cref{lemma:LCSembedding} to the conjugacy relations in $\mcR_0$, we see that 
\begin{enumerate}[(a)]
    \item the $*$-homomorphism from $\C^*\ang{\mcW_{flat}}\arr \C^*\ang{\wtd{\mcX}}$ sending $x\mapsto x$ for all $x\in\mcW_{flat}$ descends to an embedding from $\mcA_{flat}(\msB)\hookrightarrow \msA(\wtd{\msB})$,
    \item any tracial state on $\mcA_{flat}(\msB)$ extends to a tracial state on $\msA(\wtd{\msB})$, and
    \item any $r\in \mcR_{flat}$ has an $\mcR(\wtd{\msB})$-decomposition in $\C\Z_2^{*\wtd{\mcX}}$ of size $\leq 1750$.
\end{enumerate}

Parts (1) and (2) of \Cref{thm:blackbox2} follow straightforwardly from parts (1) and (2) of \Cref{lemma:flatBCS} and (a) and (b) above. 

    For part (3), suppose $\beta\in \C^*\ang{\mcX}$ is trivial in $\mcA^{\mcX}_{\text{nest}}(\msB)$ and has an $\mcR^{\mcX}_{\text{nest}}(\msB)$-decomposition in $\C\Z_2^{\mcX}$ of size $\Lambda$. By part (3) of \Cref{lemma:flatBCS}, $\beta$ has an $\mcR_{\text{flat}}(\msB)$-decomposition in $\C\Z_2^{*\mcW_{\text{flat}}}$ of size $\leq 9M^2\ell^2 \Lambda$. Note that $\norm{r}_{\C\Z_2^{*\mcW_{flat}}}\geq 1$ for all $r\in \mcR_{flat}(\msB)$ and that $\norm{r}_{\C\Z_2^{*\wtd{\mcX}}}\leq 2$ for all $r\in\mcR(\wtd{\msB})$. It follows from \Cref{lemma:Rdecomp} and part (c) above that $\beta$ has an $\mcR(\wtd{\msB})$-decomposition in $\C\Z_2^{*\wtd{\mcX}}$ of size
    \begin{equation*}
        (1+2)\cdot 1750\cdot 9 M^2\ell^2\Lambda=47250M^2\ell^2\Lambda.
    \end{equation*}
Rounding $47250$ to $2^{16}$ completes the proof.    
\end{proof}

\subsection{Proof of Theorem \ref{thm:Lfamily}}

Fix an RE set $\mcL\subseteq \N$. We first recall the construction of the family of $*$-algebras $\mcA_\mcL(m),m\in\N$ from \cite{MSZ23}. We start with a finitely-presented group $H_\mcL=\ang{\mcX_H:\mcR_H}$, where $x^2=1$ on $H_\mcL$ for all $x\in\mcX_H$, and $\mcX_H$ contains variables $\{J,X,Z,S_1,S_2,T_1,T_2,W_1,W_2 \}$. We use $S,T,W$ to denote the words $S_1S_2,T_1T_2,W_1W_2$, respectively. For all $m\in\N$ and $i\in\Z$, we use $X_{mi}$ and $Z_{mi}$ to denote the words $S^i W^m X W^{-m} S^{-i}$ and $T^i W^m Z W^{-m} T^{-i}$, respectively. We immediately see that $X_{mi}$ and $Z_{mi}$ are nested conjugacy monomials over $\mcX_H$ of depth $2m+2i$. Moreover, in $H_\mcL$, $J$ commutes with $S,T,W,X$ and $Z$. So $J$ commutes with all words $X_{mi}$ and $Z_{mi}$ in $H_\mcL$.

We  have another set of variables
\begin{equation*}
    \mcX_0:=\{U_1,U_2,\widetilde{X},\widetilde{Z},O_P,O_Q\},
\end{equation*}
and we use the same convention as in \cite{MSZ23}
\begin{align*}
    U :=U_1U_2,\quad P:=\frac{1-O_P}{2}, \text{ and } Q :=\frac{1-O_Q}{2}
\end{align*}
in $\C^*\ang{\mcX_0}$.

For every $m\in\N$, let $\mcA_{\mcL}(m)$ be the finitely
presented $*$-algebra generated by $\mcX:=\mcX_H \cup \mcX_0$, subject to the
relations
\begin{enumerate}[(R1)]
\setcounter{enumi}{-1}
    \item $x^* x = x x^* = x^2 = 1$ for all $x \in \mcX_H \cup \mcX_0$,
    \item $r=1$ for all $r \in \mcR_H$, 
    \item $[U,X]=[U,Z]=[U,S]=[U,T]=[U,J]=[Q,X]=[Q,Z]=[Q,S]=[Q,T]=[Q,J]=0$,
    \item $[\widetilde{X},Q] = \widetilde{X}Q-X_{m,0}Q = 0$ and $[\widetilde{Z},Q]=\widetilde{Z}Q-Z_{m,0}Q = 0$,
    \item $U\widetilde{X}U^*-S\widetilde{X}S^* = U\widetilde{Z}U^*-T\widetilde{Z}T^* = 0$,
    \item $[P,Q]=0$, and
    \item $(P+\widetilde{X}P\widetilde{X}-UPU^*)(\Id+J\widetilde{X}\widetilde{Z}\widetilde{X}\widetilde{Z})=0$.
\end{enumerate}
Note that relations (R0) imply that all the generators are order-two unitaries, and relations (R1) are the algebraic form of group relations in $\mcR_H$.
We use $\mcR_m$ to denote the relations in (R2)-(R6).\footnote{In \cite{MSZ23}, $\mcR_m$ refers to the relations (R1)-(R6). We modify this convention here in order to address group relations (R1) separately.} So $\mcA_\mcL(m)=\C\Z_2^{*\mcX}/\ang{(R1)\cup\mcR_m}$.

For all $n\geq 0$, let 
\begin{equation*}
    \widetilde{P}_n:=QU^nPU^{-n}Q \text{ and } \widetilde{X}_n:=U^n\widetilde{X}U^{-n}
\end{equation*}
in $\C^*\ang{\mcX}$. Note that $\wtd{X}_n^2 = 1$ in $\C \Z_2^{*\mcX}$. The following proposition is the halting case of \cite[Proposition 4.11]{MSZ23}.
\begin{proposition}\label{prop:sizeofRdecomp}
    There are positive integers $C$ and $k$ such that for any $m \notin \mcL$ and $n\geq 0$, 
\begin{equation*}
    \wtd{P}_n + \wtd{X}_n \wtd{P}_n \wtd{X}_n-\wtd{P}_{n+1}
\end{equation*}
    is trivial in $\mcA_{\mcL}(m) = \C \Z_2^{*\mcX} /
    \ang{(R1)\cup\mcR_m}$, and has an $(R1)\cup\mcR_m$-decomposition in $\C \Z_2^{*\mcX}$ of size $\leq C \big((n+1)m\big)^k$. 
\end{proposition}

Now we are ready to modify $\mcA_\mcL(m)$ to get a nested conjugacy BCS algebra $\wtd{\mcA}_\mcL(m)$. To start, applying \Cref{lemma:LCSembedding} to the group $H_\mcL=\ang{\mcX_H\colon\mcR_H}$ yields the following:
\begin{proposition}\label{prop:Hembedding}
    There is a computable boolean constraint system $\msB_H:=(\wtd{\mcX}_H,\{\msC_i^{(H)}=(\msU_i^{(H)},\msR_i^{(H)})\}_{i=1}^h)$ with $\mcX_H\subseteq \wtd{\mcX}_H$ and computable positive integers $M_H,C_H$ such that 
    \begin{enumerate}[(H1)]
    \item the natural $*$-homomorphism $\C^*\ang{\mcX_H}\arr\C^*\ang{\wtd{\mcX}_H}$ descends to an embedding $\C H_\mcL\hookrightarrow \mcA(\msB_H)$,
    \item $\mcA(\msB_H)$ has a tracial state $\tau_H$,
    \item for every $r\in\mcR$, $1-r$ has an $\mcR(\msB)$-decomposition in $\C\Z_2^{*\wtd{\mcX}}$ of size $\leq C_H$, and
    \item $\abs{\msU_i^{(H)}}\leq M_H$ for all $1\leq i\leq h$.
    \end{enumerate}
\end{proposition}
\begin{proof}
Parts (H1), (H2), (H3), and the computability of $C_H$ follow straightforwardly from \Cref{lemma:LCSembedding}. Let $M_H=\max_{i}\abs{\msU_i^{(H)}}$. Then part (H4) holds. The computability of $M_H$ follows from the computability of $\msB_H$.
\end{proof}

Next, we introduce a new variable $x_D$ and make the convention that $D=\frac{1-x_D}{2}$. For every $m\in\N$, let $\wtd{\mcA}_\mcL(m)$ be the finitely
presented $*$-algebra generated by $\wtd{\mcX}:=\wtd{\mcX}_H\cup\mcX_0\cup\{x_D\}$, subject to the
relations
\begin{enumerate}[($\wtd{\text{R}}$1)]
\setcounter{enumi}{-1}
    \item $x^* x = x x^* = x^2 = 1$ for all $x \in \wtd{\mcX}$,
    \item relations in $\mcR(\msB_H)$, 
    \item $[U,X]=[U,Z]=[U_1,S]=[U_2,S]=[U_1,T]=[U_2,T]=[U,J]=[Q,X]=[Q,Z]=[Q,S]=[Q,T]=[Q,J]=0$, i.e., comparing to (R2), ($\wtd{\text{R}}$2) replaces the commutation relation $[U,S]$ with two commutation relations $[U_1,S]$ and $[U_2,S]$, and replaces $[U,T]$ with $[U_1,T]$ and $[U_2,T]$, 
    \item relations in (R3),
    \item relations in (R4),
    \item $[P,Q]=0$ and $D=PQ$, i.e., adding the relation $D=PQ$ to (R5),
    \item $(P+\widetilde{X}P\widetilde{X}-UPU^*)(\Id+J\widetilde{X}\widetilde{Z}\widetilde{X}\widetilde{Z})=0$, and $P,\widetilde{X}P\widetilde{X},UPU^*,J,\widetilde{X}\widetilde{Z}\widetilde{X},\widetilde{Z}$ mutually commute.
\end{enumerate}
We use $\wtd{\mcR}_m$ to denote relations ($\wtd{\text{R}}$2)-($\wtd{\text{R}}$6). So $\wtd{\mcA_\mcL}(m)=\C\Z_2^{*\wtd{\mcX}}/\ang{\mcR(\msB_H)\cup\wtd{\mcR}_m}$.

\begin{proposition}\label{prop:nestedBCS}
    For every $m\in \N$, $\mcR(\msB_H)\cup\wtd{\mcR}_m=\mcR_{next}^{\wtd{\mcX}}(\wtd{\msB}_m)$ for some boolean constraint system $\wtd{\msB}_m=(\mcW_m,\{\wtd{\msC}_i^{(m)}=(\wtd{\msU}_i^{(m)},\wtd{\msR}_i^{(m)})\})$, where $\mcW_m\subseteq \mcN^{var}_{2m}(\wtd{\mcX})$ are nested conjugacy variables of depth $\leq 2m$, and $\max_i\abs{\wtd{\msU}_i^{(m)}}\leq M_H+6$.
\end{proposition}

\begin{proof}
    Let $\Psi_{\wtd{\mcX}}$ be the bijection from $\mcN^{var}(\mcX)\arr\mcN^{mon}(\mcX)$ as defined in \Cref{def:nestedvariable}. Let
    \begin{align*}
        \mcW_m:=\wtd{\mcX}\cup\Psi_{\wtd{\mcX}}^{-1}\big(\{&U_1U_2XU_2U_1,U_1U_2ZU_2U_1,S_1S_2U_1S_2S_1,S_1S_2U_2S_2S_1,\\
        &T_1T_2U_1T_2T_1,T_1T_2U_2T_2T_1, U_1U_2JU_2U_1, S_1S_2O_QS_2S_1, \\
        &T_1T_2O_QT_2T_1 , X_{m0},Z_{m0}, U_1U_2\wtd{X}U_2U_1, S_1S_2\wtd{X}S_2S_1,\\
        &U_1U_2\wtd{Z}U_2U_1, S_1S_2\wtd{Z}S_2S_1,\wtd{X}O_P\wtd{X},U_1U_2O_PU_2U_1,\wtd{X}\wtd{Z}\wtd{X}   \}\big)
    \end{align*}
    So each nested conjugacy variable in $\mcW_m$ has depth $\leq 2m$. Next we explicitly construct the boolean constraint system $\wtd{\msB}_m$ over $\mcW_m$.
    \begin{itemize}
        \item $\mcR(\msB_H)$ are already BCS relations associated with constraints $\{\msC_i^{(H)}\}_{i=1}^h$. Hence we have ($\wtd{\text{R}}$1)$=\bigcup_{i=1}^h\mcR_{nest}^{\wtd{\mcX}}(\msC_i^{(H)})$.
        \item All types of relations in ($\wtd{\text{R}}$2) are covered in \Cref{example:BCS} and \Cref{example:nestedBCS}. For instance, $[U,X]=0$ is the nested conjugacy BCS relation associated with the constraint $\Psi_{\wtd{\mcX}}^{-1}(U_1U_2XU_2U_1)=X$; $[Q,X]=0$ is the BCS relation associated with the constraint $True(O_Q,X)$. Hence we have 12 constraints $\{\msC_i\}_{i=1}^{12}$ with contexts in $\mcW_m$ such that 
        \begin{equation*}
    (\wtd{\text{R}}2)=\bigcup_{i=1}^{12}\mcR_{nest}^{\wtd{\mcX}}(\msC_i).
        \end{equation*}
        
        \item $(\wtd{\text{R}}3)=\mcR_{nest}^{\wtd{\mcX}}(\msC_{13})\cup\mcR_{nest}^{\wtd{\mcX}}(\msC_{14})$, where
        \begin{itemize}
            \item $\msC_{13}: \Psi_{\wtd{\mcX}}^{-1}(X_{m0})\wedge O_Q=\wtd{X}\wedge O_Q$, and
            \item $\msC_{14}: \Psi_{\wtd{\mcX}}^{-1}( Z_{m0})\wedge O_Q=\wtd{Z} \wedge O_Q$.
        \end{itemize}
        \item $(\wtd{\text{R}}4)=\mcR_{nest}^{\wtd{\mcX}}(\msC_{15})\cup\mcR_{nest}^{\wtd{\mcX}}(\msC_{16})$, where
        \begin{itemize}
            \item $\msC_{15}: \Psi_{\wtd{\mcX}}^{-1}(U_1U_2\wtd{X}U_2U_1)=\Psi_{\wtd{\mcX}}^{-1}(S_1S_2\wtd{X}S_2S_1)$, and
            \item $\msC_{16}: \Psi_{\wtd{\mcX}}^{-1}(U_1U_2\wtd{Z}U_2U_1)=\Psi_{\wtd{\mcX}}^{-1}(T_1T_2\wtd{X}T_2T_1)$.
        \end{itemize}
        
        \item $(\wtd{\text{R}}5)=\mcR_{nest}^{\wtd{\mcX}}(\msC_{17})$, where $\msC_{17}:O_D=O_{P}\wedge O_{Q}$.
        
        \item $(\wtd{\text{R}}6)=\mcR_{nest}^{\wtd{\mcX}}(\msC_{18})$, where $\msC_{18}$ is the boolean constraint with contexts 
        \begin{equation*}
            O_P,\Psi_{\wtd{\mcX}}^{-1}(\wtd{X}O_P\wtd{X}),\Psi_{\wtd{\mcX}}^{-1}(U_1U_2O_PU_2U_1),J,\Psi_{\wtd{\mcX}}^{-1}(\wtd{X}\wtd{Z}\wtd{X} ),\wtd{Z},
        \end{equation*}
and unsatisfying \footnote{We specify the \emph{unsatisfying} assignments here because it is easier to derive the corresponding constraint relation based on \Cref{def:bcs_alg}. The satisfying assignments are just the complement of these.} assignments $A_{123}\times A_{456}$ for
\begin{align*}
    &A_{123}=\{(-1,+1,+1),(+1,-1,+1),(+1,+1,-1),(-1,-1,+1),(-1,-1,-1) \},\\
    &A_{456}=\{(+1,+1,+1),(+1,-1,-1),(-1,+1,-1),(-1,-1,+1) \}.
\end{align*}
    \end{itemize}
Let 
\begin{equation*}
    \wtd{\msB}_m:=\left( \mcW_m, \left(\bigcup_{i=1}^h\msC_i^{(H)}\right) \cup \left( \bigcup_{i=1}^{18}\msC_i  \right)     \right).
\end{equation*}
It follows that 
\begin{equation*}
    \wtd{\mcR}_m=\left(\bigcup_{i=1}^h\mcR_{nest}^{\wtd{\mcX}}(\msC_i^{(H)})\right) \cup \left( \bigcup_{i=1}^{18}\mcR_{nest}^{\wtd{\mcX}}(\msC_i)  \right)=\mcR_{nest}^{\wtd{\mcX}}(\wtd{\msB}_m).
\end{equation*}
The deepest nested conjugacy variable in $\mcW_m$ is $\Psi_{\wtd{\mcX}}^{-1}(X_{m0})$ (and $\Psi_{\wtd{\mcX}}^{-1}(Z_{m0})$) of depth $2m$. Hence $\mcW_m\subseteq \mcN^{var}_{2m}(\wtd{\mcX})$. Every constraint $\msC_i^{(H)},1\leq i\leq h$ has context size $\leq M_H$, and every constraint $\msC_i,1\leq i\leq 18$ has context size $\leq 6$. We conclude that every constraint in $\wtd{\msB}_m$ has context size $\leq M_H+6$.
\end{proof}

\Cref{prop:nestedBCS} says that every $\wtd{\mcA}_\mcL(m)=\C\Z_2^{*\wtd{\mcX}}/\ang{\mcR(\msB_H)\cup\wtd{\mcR}_m}$ is the nested conjugacy BCS algebra $\mcA_{nest}^{\wtd{\mcX}}(\wtd{\msB}_m)$ associated with some boolean constraint system $\wtd{\msB}_m$. Applying \Cref{thm:blackbox2} to all $\{\wtd{\mcA}_\mcL(m)\}_{m\in\N}$ yields a family of boolean constraint systems $\{\hat{\msB}_m=(\hat{\mcX}_m,\{\hat{\msC}_i^{(m)}\})\}_{m\in\N}$ such that for every $m\in\N$,
\begin{enumerate}
    \item $\wtd{\mcX}\subseteq \hat{\mcX}_m$, and the $*$-homomorphism $\C^*\ang{\wtd{\mcX}} \to \C^*\ang{\hat{\mcX}_m}$ sending $x \mapsto x$ for all  $x \in \wtd{\mcX}$  descends to an embedding from $\mcA^{\wtd{\mcX}}_{\text{nest}}(\wtd{\msB}_m)\hookrightarrow \mcA(\hat{\msB}_m)$,
    \item any tracial state on $\mcA^{\wtd{\mcX}}_{\text{nest}}(\wtd{\msB}_m)$ extends to a tracial state on $\mcA(\hat{\msB}_m)$, and
     \item if $\beta\in \C^*\ang{\wtd{\mcX}}$ has an $\mcR^{\wtd{\mcX}}_{\text{nest}}(\wtd{\msB}_m)$-decomposition in $\C\Z_2^{\wtd{\mcX}}$ of size $\Lambda$, then $\beta$ has an $\mcR(\hat{\msB}_m)$-decomposition in $\C\Z_2^{*\wtd{\mcX}}$ of size $\leq 2^{16}(M_H+6)^2m^2\Lambda$.
\end{enumerate}

We now show that $\{\hat{\msB}_m\}_{m\in\N}$ is an $\mcL$-family, thereby establishing \Cref{thm:Lfamily}.The proof is based on the following propositions.

\begin{proposition}\label{prop:final1}
    If $m\in\mcL$ then there is a tracial state $\tau$ on $\wtd{\mcA}_\mcL(m)$ such that $\tau(D)>0$.
\end{proposition}

The proof follows from the representation outlined in Lemma 4.10 of \cite{MSZ23}. However, we remark that the representation is required to satisfy a few slightly different relations. As such, we repeat the presentation of the representation for completeness, as it will be useful for the reader to see how it satisfies the BCS algebra relations.

\begin{proof}
Suppose that $m\in \mcL$ and that $M$ is the Turing machine that halts and accepts on the $n$th step upon being given input $m$. Let $\tau_0$ be the tracial state on $\mcA(\msB_H)$ in (H2) of \cref{prop:Hembedding} coming from the canonical trace on $\C H_\mcL$. Denote the GNS triple associated with $\tau_0$ by $(\mcH_0,\pi,|\nu\rangle)$. Denote by $\mcH_1$ the Hilbert space $\ell^2(\Z_{n+1})$ with canonical basis $\{|i\rangle :i\in \Z_{n+1}\}$, and let $E_{i,i}$ the rank-one projection onto $|i\rangle$ for every $i\in \Z_{n+1}$. Similarly, define $L$ to be the left cyclic shift operator on $\ell^2(\Z_{n+1})$ taking $|i\rangle\mapsto |i+1\rangle$. Let $\wtd{\mcH}=\mcH_0\otimes \mcH_1$, and define a *-representation $\wtd{\pi}:\C^*\langle \mcX_H\cup \mcX_0\cup \{x_D\}\rangle\to \mcB(\wtd{\mcH}\otimes \C^2)$ by
\begin{align*}
    \wtd{\pi}(x)=\begin{pmatrix}
        \pi(x)\otimes \Id_{\mcH_1} & \\
        & \pi(x)\otimes \Id_{\mcH_1}
    \end{pmatrix},
\end{align*}
for all $x\in \mcX_H$, and
\begin{align*}
    \wtd{\pi}(U_1)&=\begin{pmatrix}
        & \Id_{\mcH_0}\otimes L \\
        \Id_{\mcH_0}\otimes L^{-1} &
    \end{pmatrix}, \; \wtd{\pi}(U_2)=\begin{pmatrix}
        & \Id_{\wtd{\mcH}} \\
        \Id_{\wtd{\mcH}} &
    \end{pmatrix},\\
    \wtd{\pi}(\wtd{X})&=\begin{pmatrix}
        \sum_{i=0}^n \pi(X_{mi})\otimes E_{-i,-i} & \\
        &  \sum_{i=0}^n \pi(X_{mi})\otimes E_{i,i}
    \end{pmatrix},\\
    \wtd{\pi}(\wtd{Z})&=\begin{pmatrix}
        \sum_{i=0}^n \pi(Z_{mi})\otimes E_{-i,-i} & \\
        &  \sum_{i=0}^n \pi(Z_{mi})\otimes E_{i,i}
    \end{pmatrix},\\
    \wtd{\pi}(Q)&=\begin{pmatrix}
        \Id_{\mcH_0}\otimes E_{0,0} & \\
        &  \Id_{\mcH_0}\otimes E_{0,0}
    \end{pmatrix},\\
    \wtd{\pi}(P)&=\begin{pmatrix}
        \sum_{i=0}^n \pi(P_{i})\otimes E_{-i,-i} & \\
        &  \sum_{i=0}^n \pi(P_{i})\otimes E_{i,i}
    \end{pmatrix},
\end{align*}
where $P_i=(\frac{1-J}{2})\prod_{i\leq j \leq n-1}\left(\frac{1-Z_{mj}}{2}\right)$ for all $0\leq i \leq n-1$, and $P_n:=\left(\frac{1-J}{2}\right)$. Lastly,
\begin{align*}
        \wtd{\pi}(D)&=\wtd{\pi}(PQ)=\begin{pmatrix}
        \pi(P_0)\otimes E_{0,0} & \\
        &  \pi(P_0)\otimes E_{0,0}
    \end{pmatrix}
\end{align*}

We leave it to the reader to verify that the relations ($\tilde{\text{R}}0$) hold. Next, we observe that the relations ($\tilde{\text{R}}1$) hold via the embedding theorem \cref{prop:Hembedding}. For the relations ($\tilde{\text{R}}3$), and ($\tilde{\text{R}}4$), we refer the reader to the proof in \cite{MSZ23} of Lemma 4.10. For ($\tilde{\text{R}}2$) we remark that although the representation is defined for $\mcA_\mcL$, we have that $\wtd{\pi}(U_i)\wtd{\pi}(x)=\wtd{\pi}(x)\wtd{\pi}(U_i)$ for all $x\in \mcX_H$. In particular, it holds when $x=S$ and $x=T$. For the remaining relations of ($\tilde{\text{R}}2$) we refer to the proof in \cite{MSZ23}. The additional relation in ($\tilde{\text{R}}5$) holds by the definition of $\wtd{\pi}(P)$ and $\wtd{\pi}(Q)$. For ($\tilde{\text{R}}6$) the proof in \cite{MSZ23} shows that $\wtd{\pi}\left((P+\widetilde{X}P\widetilde{X}-UPU^*)(\Id+J\widetilde{X}\widetilde{Z}\widetilde{X}\widetilde{Z}\right)=0$.

For the remaining commutators in ($\tilde{\text{R}}6$), the relations (G1), (G2), and (G3) for the finite presentation of $H_\mcL$ in \cite{MSZ23} state that
\begin{enumerate}
    \item[(G1)] $J$ commutes with $X_{mi}$ and $Z_{mi}$ for all $0\leq i\leq n$,
    \item[(G2)] $X_{mi}Z_{mi}=JZ_{mi}X_{mi}$ for all $0\leq i\leq n-1$ and  $X_{mn}Z_{mn}=Z_{mn}X_{mn}$ 
    \item[(G3)] $[X_{mi},Z_{mj}]=[X_{mi},X_{mj}]=[Z_{mi},Z_{mj}]=0$ for all $i\neq j$.
\end{enumerate}
These relations imply that for every $0\leq i\leq n$, the operators $ \pi(P_i)$, $\pi(X_{mi}P_iX_{mi})$, $\pi(P_{i+1})$, $\pi(J)$, $\pi(X_{mi}Z_{mi}X_{mi}) $, and $\pi(Z_{mi})$ mutually commute. Hence, we obtain that $\wtd{\pi}(P)$, $\wtd{\pi}(\wtd{X}P\wtd{X})$, $\wtd{\pi}(UPU^*)$, $\wtd{\pi}(J)$, $\wtd{\pi}(\wtd{X}\wtd{Z}\wtd{X})$, and $\wtd{\pi}(\wtd{Z})$ mutually commute.

It follows that $\widetilde{\pi}$ descends to a representation of $\wtd{\mcA_{\mcL}}(m)$ on $\wtd{\mcH}\otimes \C^2$. In particular, the image of $\wtd{\pi}$ is contained in $M_2\left(\pi(\C H_\mcL)\otimes \mcB(\mcH_1)\right)\iso \pi(\C H_\mcL)\otimes M_{2(n+1)}(\C)$. Hence, if $\text{tr}(\cdot)$ is the unique normalized trace on $M_{2(n+1)}(\C)$, then $\tau:=\tau_0\otimes \tr$ is a tracial state on $\wtd{\mcA_{\mcL}}(m)$, and by the proof in \cite{MSZ23}, we obtain
\begin{equation*}
    \tau\left(\wtd{\pi}(D)\right)=\frac{1}{(n+1)}\tau_0(P_0)=\frac{1}{(n+1)2^{n+1}}>0,
\end{equation*}
as desired.
\end{proof}

\begin{proposition}\label{prop:final2}
     There are constants $\wtd{C}>0$ and $\wtd{k}\geq 2$ such that for any $m \not\in\mcL$ and $n\geq 0$, $ \wtd{P}_n + \wtd{X}_n \wtd{P}_n \wtd{X}_n-\wtd{P}_{n+1}$
    is trivial in $\wtd{\mcA}_{\mcL}(m)$, and has an $\wtd{\mcR}_m$-decomposition in $\C \Z_2^{*\wtd{\mcX}}$ of size $\leq \wtd{C} \big((n+1)m\big)^{\wtd{k}}$. 
\end{proposition}

\begin{proof}
Let $C$ and $k$ be as in \Cref{prop:sizeofRdecomp}, and fix $m\in\mcL$ and $b\geq 0$. By (H1) of \Cref{prop:Hembedding}, the natural $*$-homomorphism $\C^*\ang{\mcX_H\cup\mcX_0}\arr \C^*\ang{\wtd{\mcX}_H\cup\mcX_0}$ descends to an embedding $\mcA_\mcL(m)\hookrightarrow \C\Z^{*\wtd{\mcX}_H\cup\mcX_0}/\ang{\mcR(\msB_H)\cup\mcR_m}$. Then by \Cref{lemma:Rdecomp} and (H3) of \Cref{prop:Hembedding}, $\wtd{P}_n + \wtd{X}_n \wtd{P}_n \wtd{X}_n-\wtd{P}_{n+1}$ has an $\mcR(\msB_H)\cup\mcR_m$-decomposition in $\C\Z^{*\wtd{\mcX}_H\cup\mcX_0}$ of size $\leq 2C_GC\big((n+1)m\big)^k$. Here we used that $\norm{1-r}_{\C\Z^{*\mcX_H\cup\mcX_0}}\geq 2$ for all $r\in\mcR_H$ and $\norm{r}_{\C\Z^{*\wtd{\mcX}_H\cup\mcX_0}}\leq 2$ for all $r\in\mcR(\msB_H)$.

Note that for any $r\in \mcR_m$ that is not in $\wtd{\mcR}_m$, there exist $r_1,r_2\in \wtd{\mcR}_m$ and $w_1,w_2\in \wtd{\mcX}$ such that $r= w_1r_1+r_2w_2$ in $\C\Z_2^{*\wtd{\mcX}}$. For instance,
\begin{equation*}
    UX-XU=U_1(U_2X-XU_2)+(U_1X-XU_1)U_2. 
\end{equation*}
This implies $\wtd{P}_n + \wtd{X}_n \wtd{P}_n \wtd{X}_n-\wtd{P}_{n+1}$ has an $\mcR(\msB_H)\cup\wtd{\mcR}_m$-decomposition in $\C\Z^{*\wtd{\mcX}}$ of size $\leq 6C_GC\big((n+1)m\big)^k$. Then it follows from \Cref{thm:blackbox2} that $\wtd{P}_n + \wtd{X}_n \wtd{P}_n \wtd{X}_n-\wtd{P}_{n+1}$ is trivial in $\mcA(\hat{\msB}_m)$ and has an $\mcR(\hat{\msB}_m)$-decomposition in $\C\Z_2^{\hat{\mcX}_m}$ of size $\leq 2^{16}(M_H+6)^2m^2\cdot 6C_GC\big((n+1)m\big)^k$. Taking $\wtd{C}:=6C_GC 2^{16}(M_H+6)^2$ and $\wtd{k}:=k+2$ completes the proof.
\end{proof}

Now we are ready to prove \Cref{thm:Lfamily} by showing that $
\{\hat{\msB}_m\}$ is an $\mcL$-family.

\begin{proof}[Proof of \Cref{thm:Lfamily}]
  Hypothesis (2) of \Cref{def:Lfamily} follows directly from \Cref{prop:final1}. Now suppose $m\not\in \mcL$. Let $\mcS$ be an $\epsilon$-perfect strategy for $\hat{\msB}_m$. It follows from \Cref{prop:near-perfect} that there is a constant $T_{\hat{\msB}_m}$, a $\left(T_{\hat{\msB}_m}\cdot\epsilon,\hat{\mcX}_m\right)$-synchronous state $f$ on $\C\Z_2^{*\hat{\mcX}}\otimes \C\Z_2^{*\hat{\mcX}}$ such that $\varphi_\mcS=f\circ \iota$, where $\iota:\C\Z_2^{*\mcX}\hookrightarrow \C\Z_2^{*\mcX}\otimes \C\Z_2^{*\mcX} $ is the left inclusion, and that $\varphi_S$ is a $(T_{\hat{\msB}_m}\cdot\epsilon,\mcR(\wtd{\msB}_m))$-state. Let $\Gamma_m$ be any integer $\geq T_{\hat{\msB}_m}(\wtd{C}+25)\sum_{n=1}^{\infty}\tfrac{n^{\wtd{k}}}{2^{n/2}}$. Then by \cite[Proposition 6.3]{MSZ23} and its proof, \Cref{prop:final2} implies that 
  \begin{equation*}
      \norm{PQ}_{\varphi_\mcS}\leq \Gamma_m m^{\wtd{k}}\sqrt{\epsilon}.
  \end{equation*}
By part (3) of \Cref{thm:blackbox2}, the relation $D-PQ$ in ($\wtd{\text{R}}$4) has an $\mcR(\hat{\msB}_m)$-decomposition of size $\leq 2^{16}(M_H+6)(2m)^2$. It follows from \Cref{lemma:Rbound} that 
\begin{equation*}
    \norm{D-PQ}_{\varphi_\mcS}\leq 2^{18}(M_H+6)m^2 \sqrt{\epsilon}.
\end{equation*}
Hence 
\begin{align*}
    \norm{D}_{\varphi_\mcS}&\leq \Gamma_m m^{\wtd{k}}\sqrt{\epsilon} + 2^{18}(M_H+6)m^2 \sqrt{\epsilon}\\
    &\leq (\Gamma_m+2^{18}(M_H+6))m^{\wtd{k}}\sqrt{\epsilon}.
\end{align*}
Let $C_m:=(\Gamma_m+2^{18}(M_H+6))^2 m^{2\wtd{k}}$. We have that $\varphi_\mcS(D)\leq C_m \epsilon$. So $\{\hat{\msB}_m \}_{m\in\N}$ and $\{C_m\}_{m\in\N}$ satisfy hypothesis (3) of \Cref{def:Lfamily}. From the construction of $C_m$, hypothesis (1) also follows.
\end{proof}

\bibliographystyle{alpha}
\bibliography{dec}

\end{document}